\documentclass[a4paper,11pt]{article}
\usepackage[margin = 1in]{geometry}
\usepackage[utf8]{inputenc}
\usepackage[T1]{fontenc}
\usepackage[dvipsnames]{xcolor}
\usepackage{graphicx}
\usepackage{amsmath, amssymb, amstext, mathtools}
\usepackage{amsthm}

\usepackage{tikz}
\usepackage{pgfplots}

\usepackage{float}
\usepackage{hyperref}
\usepackage{paralist}
\usepackage{enumitem}
\usepackage[titlenumbered, linesnumbered, ruled]{algorithm2e}

\usepackage{subcaption}

\allowdisplaybreaks

\definecolor[named]{urlblue}{cmyk}{1,0.58,0,0.21}

\hypersetup{
 breaklinks=true,
 colorlinks=true,
 citecolor=purple!70!blue!60!black,
 linkcolor=purple!70!blue!90!black,
 urlcolor=urlblue,
 pdflang={en},
 pdftitle={Isomorphism Testing Parameterized by Genus and Beyond},
 pdfauthor={Daniel Neuen}
}

\usetikzlibrary{patterns,arrows,decorations.pathreplacing,decorations.pathmorphing,calc}

\tikzstyle{vertex}=[circle,fill=white,draw=black,minimum size=8pt,inner sep=0pt]

\pgfdeclarelayer{background}
\pgfdeclarelayer{foreground}
\pgfsetlayers{background,main,foreground}

\definecolor[named]{mGray}{rgb}{0.31,0.31,0.33}

\definecolor[named]{mBlue}{RGB}{122,148,229}
\definecolor[named]{mTurquoise}{RGB}{0,152,161}
\definecolor[named]{mGreen}{RGB}{87,171,39}
\definecolor[named]{mYellow}{RGB}{254,200,18}
\definecolor[named]{mRed}{RGB}{204,7,30}

\definecolor[named]{mOrange}{RGB}{255,128,0}
\definecolor[named]{mViolet}{RGB}{97,33,88}
\definecolor[named]{mDarkBlue}{RGB}{0,84,159}
\definecolor[named]{mLightGreen}{RGB}{189,205,0}

\newtheorem{theorem}{Theorem}[section]
\newtheorem{lemma}[theorem]{Lemma}
\newtheorem{corollary}[theorem]{Corollary}

\theoremstyle{definition}
\newtheorem{definition}[theorem]{Definition}

\theoremstyle{remark}
\newtheorem{remark}[theorem]{Remark}
\newtheorem{claim}[theorem]{Claim}

\newenvironment{claimproof}{\begin{proof}}{\end{proof}}

\newcommand{\WL}[2]{\chi_{\sf WL}^{#1}\![#2]}
\newcommand{\WLit}[3]{\chi_{(#2)}^{#1}[#3]}

\newcommand{\tWL}[3]{\chi_{(#1,#2)\text{-}{\sf WL}}[#3]}

\newcommand{\CA}{{\mathcal A}}

\newcommand{\CE}{{\mathcal E}}

\newcommand{\CG}{{\mathcal G}}
\newcommand{\CH}{{\mathcal H}}

\newcommand{\CM}{{\mathcal M}}

\newcommand{\CO}{{\mathcal O}}
\newcommand{\CP}{{\mathcal P}}
\newcommand{\CQ}{{\mathcal Q}}

\newcommand{\CS}{{\mathcal S}}

\newcommand{\NN}{{\mathbb N}}

\DeclareMathOperator{\im}{im}

\DeclareMathOperator{\id}{id}

\DeclareMathOperator{\Iso}{Iso}
\DeclareMathOperator{\Aut}{Aut}
\DeclareMathOperator{\Sym}{Sym}

\DeclareMathOperator{\mgamma}{\widehat{\Gamma}}

\DeclareMathOperator{\dist}{dist}
\DeclareMathOperator{\cl}{cl}
\DeclareMathOperator{\parent}{par}

\newcommand{\orcid}[1]{\href{https://orcid.org/#1}{\includegraphics[height=1.8ex]{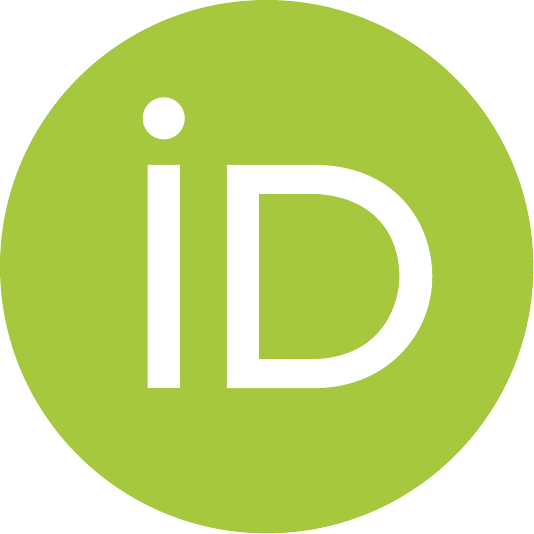}}}
\newcommand{\email}[1]{\href{mailto:#1}{\texttt{#1}}}

\title{Isomorphism Testing Parameterized by Genus and Beyond}
\author{
Daniel Neuen \orcid{0000-0002-4940-0318}\\
University of Bremen\\
\email{dneuen@uni-bremen.de}
}
\date{}

\begin{document}

\maketitle

\begin{abstract}
 We give an isomorphism test for graphs of Euler genus $g$ running in time $2^{\CO(g^4 \log g)}n^{\CO(1)}$.
 Our algorithm provides the first explicit upper bound on the dependence on $g$ for an fpt isomorphism test parameterized by the Euler genus of the input graphs.
 The only previous fpt algorithm runs in time $f(g)n$ for some function $f$ (Kawarabayashi 2015).
 Actually, our algorithm even works when the input graphs only exclude $K_{3,h}$ as a minor.
 For such graphs, no fpt isomorphism test was known before.
 
 The algorithm builds on an elegant combination of simple group-theoretic, combinatorial, and graph-theoretic approaches.
 In particular, our algorithm relies on the notion $(t,k)$-WL-bounded graphs which provide a powerful tool to combine group-theoretic techniques with the standard Weisfeiler-Leman algorithm.
 This concept may be of independent interest.
\end{abstract}

\section{Introduction}

Determining the computational complexity of the Graph Isomorphism Problem is a long-standing open question in theoretical computer science (see, e.g., \cite{Karp72}).
The problem is easily seen to be contained in NP, but it is neither known to be in PTIME nor known to be NP-complete.
In a breakthrough result, Babai \cite{Babai16} recently obtained a quasipolynomial-time algorithm for testing isomorphism of graphs (i.e., an algorithm running in time $n^{\CO((\log n)^c)}$ where $n$ denotes the number of vertices of the input graphs, and $c$ is some constant), achieving the first improvement over the previous best algorithm running in time $n^{\CO(\sqrt{n / \log n})}$ \cite{BabaiKL83} in over three decades.
However, it remains wide open whether GI can be solved in polynomial time.

In this work, we are concerned with the parameterized complexity of isomorphism testing.
While polynomial-time isomorphism tests are known for a large variety of restricted graph classes (see, e.g., \cite{Bodlaender90,GroheM15,GroheS15,HopcroftT72,Luks82,Ponomarenko91}),
for several important structural parameters such as maximum degree or rank-width, it is still unknown whether isomorphism testing is fixed-parameter tractable (i.e., whether there is an isomorphism algorithm running in time $f(k)n^{\CO(1)}$ where $k$ denotes the graph parameter in question, $n$ the number of vertices of the input graphs, and $f$ is some function).
On the other hand, there has also been significant progress in recent years.
In 2015, Lokshtanov et al.\ \cite{LokshtanovPPS17} obtained the first fpt isomorphism test parameterized by the tree-width $k$ of the input graph running in time $2^{\CO(k^5 \log k)}n^5$.
This algorithm was later improved by Grohe et al.\ \cite{GroheNSW20} to a running time of $2^{\CO(k \cdot (\log k)^c)}n^3$ (for some constant $c$).
In the same year, Kawarabayashi \cite{Kawarabayashi15} obtained the first fpt isomorphism test parameterized by the Euler genus $g$ of the input graph running in time $f(g)n$ for some function $f$.
While Kawarabayashi's algorithm achieves optimal dependence on the number of vertices of the input graphs, it is also extremely complicated and it provides no explicit upper bound on the function $f$.
Indeed, the algorithm spans over multiple papers \cite{Kawarabayashi15,KawarabayashiM08,KawarabayashiMN21} and builds on several deep structural results for graphs of bounded genus.
Both results are generalized by Lokshtanov, Pilipczuk, Pilipczuk and Saurabh \cite{LokshtanovPPS22} who provide an fpt isomorphism test parameterized by the Hadwiger number\footnote{The Hadwiger number of a graph $G$ is the maximum number $h$ such that $K_h$ is a minor of $G$.} of the input graphs running in time $f(h)n^{\CO(1)}$ for some function $f$.
Again, the algorithm is extremely complicated and the authors are even unable to give a \emph{computable} upper bound for the function $f$.
This actually means that their isomorphism test is not an fpt algorithm in a strict sense \cite{CyganFKLMPPS15,FlumG06}.

The last two results naturally raise the question of whether there are simpler algorithms that also achieve a better dependence on the parameter in question (e.g., the function $f$ is exponentially bounded).
In this work, we present such an isomorphism test for graphs of Euler genus at most $g$ running in time $2^{\CO(g^4 \log g)}n^{\CO(1)}$.
In contrast to Kawarabayashi's algorithm, our algorithm does not require any deep graph-theoretic insights, but rather builds on an elegant combination of well-established and simple group-theoretic, combinatorial, and graph-theoretic ideas.
In particular, this enables us to provide the first explicit upper bound on the dependence on $g$ for an fpt isomorphism test.
Actually, the only property our algorithm exploits is that graphs of genus $g$ exclude $K_{3,h}$ as a minor for $h \geq 4g+3$ \cite{Ringel65}.
In other words, our main result is an fpt isomorphism test for graphs excluding $K_{3,h}$ as a minor.

\begin{theorem}
 The Graph Isomorphism Problem for graphs excluding $K_{3,h}$ as a minor can be solved in time $2^{\CO(h^4 \log h)}n^{\CO(1)}$.
\end{theorem}

For the algorithm, we combine different approaches to the Graph Isomorphism Problem.
On a high-level, our algorithm follows a simple decomposition strategy which decomposes the input graph $G$ into pieces such that the interplay between the pieces is simple.
Here, the main idea is to define the pieces in such a way that, after fixing a small number of vertices, the automorphism group of $G$ restricted to a piece $D \subseteq V(G)$ is similar to the automorphism groups of graphs of maximum degree $3$.
This allows us to test isomorphism between the pieces using the group-theoretic graph isomorphism machinery dating back to Luks's polynomial-time isomorphism test for graphs of bounded maximum degree \cite{Luks82}.

In order to capture the restrictions on the automorphism group, we introduce the notion of \emph{$(t,k)$-WL-bounded graphs} which generalize so-called $t$-CR-bounded graphs.
The class of $t$-CR-bounded graphs was originally defined by Ponomarenko \cite{Ponomarenko89} and was recently rediscovered in \cite{Neuen22b,GroheNW23,Neuen22} in a series of works eventually leading to an algorithm testing isomorphism of graphs excluding $K_h$ as a topological subgraph in time $n^{\CO((\log h)^c)}$ (for some constant $c$).
Intuitively speaking, a graph $G$ is \emph{$t$-CR-bounded} if an initially uniform vertex-coloring $\chi$ can be turned into a discrete coloring (i.e., a coloring where every vertex has its own color) by repeatedly (a) applying the standard Color Refinement algorithm, and (b) splitting all color classes of size at most $t$.
We define \emph{$(t,k)$-WL-bounded} graphs in the same way, but replace the Color Refinement algorithm by the well-known Weisfeiler-Leman algorithm of dimension $k$ (see, e.g., \cite{CaiFI92,ImmermanL90}).
It seems that this natural extension of $t$-CR-bounded graphs has not been explicitly considered so far in the literature (although related ideas appear as early as \cite{Weisfeiler76}), and we start by building a polynomial-time isomorphism test for such graphs using the group-theoretic methods developed by Luks \cite{Luks82} as well as a simple extension due to Miller \cite{Miller83b}.
Actually, it turns out that isomorphism of $(t,k)$-WL-bounded graphs can even be tested in time $n^{\CO(k \cdot (\log t)^c)}$ using recent extensions \cite{Neuen22b} of Babai's quasipolynomial-time isomorphism test.
However, since we only apply these methods for $t=k=2$, there is no need for our algorithm to rely on such sophisticated subroutines.

Now, as the main structural insight, we prove that each $3$-connected graph $G$ that excludes $K_{3,h}$ as a minor admits (after fixing $3$ vertices) an isomorphism-invariant rooted tree decomposition $(T,\beta)$ such that the adhesion width (i.e., the maximal intersection between two bags) is bounded by $h$.
Additionally, each bag $\beta(t)$, $t \in V(T)$, can be equipped with a set $\gamma(t) \subseteq \beta(t)$ of size $|\gamma(t)| \leq h^4$ such that, after fixing all vertices in $\gamma(t)$, $G$ restricted to $\beta(t)$ is $(2,2)$-WL-bounded.
Given such a decomposition, isomorphisms can be computed by a simple bottom-up dynamic programming strategy along the tree decompositions.
For each bag, isomorphism is tested by first individualizing all vertices from $\gamma(t)$ at an additional factor of $|\gamma(t)|! = 2^{\CO(h^4 \log h)}$ in the running time.
Following the individualization of these vertices, our algorithm can then simply rely on a polynomial-time isomorphism test for $(2,2)$-WL-bounded graphs.
Here, we incorporate the partial solutions computed in the subtree below the current bag via a simple gadget construction.

To compute the decomposition $(T,\beta)$, we also build on the notion of $(2,2)$-WL-bounded graphs.
Given a set $X \subseteq V(G)$, we define the \emph{$(2,2)$-closure} to be the set $D = \cl_{2,2}^G(X)$ of all vertices appearing in a singleton color class after artificially individualizing all vertices from $X$, and performing the $(2,2)$-WL procedure.
As one of the main technical contributions, we can show that the interplay between $D$ and its complement in $G$ is simple (assuming $G$ excludes $K_{3,h}$ as a minor).
To be more precise, building on various properties of the $2$-dimensional Weisfeiler-Leman algorithm, we show that $|N_G(Z)| < h$ for every connected component $Z$ of $G - D$.
This allows us to choose $D = \cl_{2,2}^G(X)$ as the root bag of $(T,\beta)$ for some carefully chosen set $X$, and obtain the decomposition $(T,\beta)$ by recursion.

\medskip

The remainder of this work is structured as follows.
In the next section, we give the necessary preliminaries.
In Section \ref{sec:t-k-wl} we introduce $(t,k)$-WL-bounded graphs and provide a  polynomial-time isomorphism test for such graphs.
Next, we give a more detailed overview on our fpt algorithm in Section \ref{sec:overview}.
The tree decomposition is then computed in Sections \ref{sec:disjoint-subtrees} and \ref{sec:decomposition}.
Finally, we assemble the main algorithm in Section \ref{sec:main-algorithm}.

\section{Preliminaries}
\label{sec:preliminaries}

\subsection{Graphs}

A \emph{graph} is a pair $G = (V(G),E(G))$ consisting of a \emph{vertex set} $V(G)$ and an \emph{edge set} $E(G)$.
All graphs considered in this paper are finite and simple (i.e., they contain no loops or multiple edges).
Moreover, unless explicitly stated otherwise, all graphs are undirected.
For an undirected graph $G$ and $v,w \in V(G)$, we write $vw$ as a shorthand for $\{v,w\} \in E(G)$.
The \emph{neighborhood} of a vertex $v \in V(G)$ is denoted by $N_G(v)$.
The \emph{degree} of $v$, denoted by $\deg_G(v)$, is the number of edges incident with $v$, i.e., $\deg_G(v)=|N_G(v)|$.
For $X \subseteq V(G)$, we define $N_G[X] \coloneqq X \cup \bigcup_{v \in X}N_G(v)$ and $N_G(X) \coloneqq N_G[X] \setminus X$.
If the graph $G$ is clear from context, we usually omit the index and simply write $N(v)$, $\deg(v)$, $N[X]$ and $N(X)$.

We write $K_{\ell,h}$ to denote the complete bipartite graph on $\ell$ vertices on the left side and $h$ vertices on the right side.
A graph is \emph{regular} if every vertex has the same degree.
A bipartite graph $G=(V_1,V_2,E)$ is called \emph{$(d_1,d_2)$-biregular} if all vertices $v_i \in V_i$ have degree $d_i$ for both $i \in \{1,2\}$.
In this case $d_1 \cdot |V_1| = d_2 \cdot |V_2| = |E|$.
By a double edge counting argument, for each subset $S \subseteq V_i$, $i\in\{1,2\}$, it holds that $|S| \cdot d_i \leq |N_G(S)| \cdot d_{3-i}$.
A bipartite graph is \emph{biregular}, if there are $d_1,d_2 \in \NN$ such that $G$ is $(d_1,d_2)$-biregular.
Each biregular graph satisfies the Hall condition, i.e., for all $S \subseteq V_1$ it holds $|S| \leq |N_G(S)|$ (assuming $|V_1| \leq |V_2|$).
Thus, by Hall's Marriage Theorem, each biregular graph contains a matching of size $\min(|V_1|,|V_2|)$.

A \emph{path of length $k$} from $v$ to $w$ is a sequence of distinct vertices $v = u_0,u_1,\dots,u_k = w$ such that $u_{i-1}u_i \in E(G)$ for all $i \in [k] \coloneqq \{1,\dots,k\}$.
The \emph{distance} between two vertices $v,w \in V(G)$, denoted by $\dist_G(v,w)$, is the length of a shortest path between $v$ and $w$.
As before, we omit the index $G$ if it is clear from context.
For two sets $A,B\subseteq V(G)$, we denote by $E_G(A,B) \coloneqq \{vw\in E(G)\mid v\in A,w\in B\}$.
Also, $G[A,B]$ denotes the graph with vertex set $A\cup B$ and edge set $E_G(A,B)$.
For $A \subseteq V(G)$, we denote by $G[A] \coloneqq G[A,A]$ the induced subgraph on $A$.
Moreover, $G - A$ denotes the subgraph induced by the complement of $A$, that is, the graph $G - A \coloneqq G[V(G) \setminus A]$.
For $F \subseteq E(G)$, we also define $G - F$ to be the graph obtained from $G$ by removing all edges contained in $F$ (the vertex set remains unchanged).
A graph $H$ is a \emph{subgraph} of $G$, denoted by $H \subseteq G$, if $V(H) \subseteq V(G)$ and $E(H) \subseteq E(G)$. 
A graph $H$ is a \emph{minor} of $G$ if $H$ can be obtained from $G$ by deleting vertices and edges, as well as contracting edges.
The graph $G$ \emph{excludes $H$ as a minor} if it does not have a minor isomorphic to $H$.

A \emph{tree decomposition} of a graph $G$ is a pair $(T,\beta)$ where $T$ is a tree and $\beta\colon V(T) \rightarrow 2^{V(G)}$ ($2^{V(G)}$ denotes the power set of $V(G)$) such that
\begin{enumerate}[label = (T.\arabic*)]
 \item for every $vw \in E(G)$ there is some $t \in V(T)$ such that $v,w \in \beta(t)$, and
 \item for every $v \in V(G)$ the set $\{t \in V(T) \mid v \in \beta(t)\}$ is non-empty and connected in $T$.
\end{enumerate}
The \emph{width} of the decomposition $(T,\beta)$ is $\max_{t \in V(T)} |\beta(t)| - 1$.
Also, the \emph{adhesion width} of $(T,\beta)$ is $\max_{st \in E(T)} |\beta(t) \cap \beta(s)|$.

An \emph{isomorphism} from $G$ to a graph $H$ is a bijection $\varphi\colon V(G) \rightarrow V(H)$ that respects the edge relation, that is, for all~$v,w \in V(G)$, it holds that~$vw \in E(G)$ if and only if $\varphi(v)\varphi(w) \in E(H)$.
Two graphs $G$ and $H$ are \emph{isomorphic}, written $G \cong H$, if there is an isomorphism from~$G$ to~$H$.
We write $\varphi\colon G\cong H$ to denote that $\varphi$ is an isomorphism from $G$ to $H$.
Also, $\Iso(G,H)$ denotes the set of all isomorphisms from $G$ to $H$.
The automorphism group of $G$ is $\Aut(G) \coloneqq \Iso(G,G)$.
Observe that, if $\Iso(G,H) \neq \emptyset$, it holds that $\Iso(G,H) = \Aut(G)\varphi \coloneqq \{\gamma\varphi \mid \gamma \in \Aut(G)\}$ for every isomorphism $\varphi \in \Iso(G,H)$.

A \emph{vertex-colored graph} is a tuple $(G,\chi_V)$ where $G$ is a graph and $\chi_V\colon V(G) \rightarrow C$ is a mapping into some set $C$ of colors, called \emph{vertex-coloring}.
Similarly, an \emph{arc-colored graph} is a tuple $(G,\chi_E)$, where $G$ is a graph and $\chi_E\colon\{(u,v) \mid \{u,v\} \in E(G)\} \rightarrow C$ is a mapping into some color set $C$, called \emph{arc-coloring}.
Observe that colors are assigned to directed edges, i.e., the directed edge $(v,w)$ may obtain a different color than $(w,v)$.
We also consider vertex- and arc-colored graphs $(G,\chi_V,\chi_E)$ where $\chi_V$ is a vertex-coloring and $\chi_E$ is an arc-coloring.
Typically, $C$ is chosen to be an initial segment $[n]$ of the natural numbers.
To be more precise, we generally assume that there is a linear order on the set of all potential colors which, for example, allows us to identify a minimal color appearing in a graph in a unique way.
Isomorphisms between vertex- and arc-colored graphs have to respect the colors of the vertices and arcs.

\subsection{Weisfeiler-Leman Algorithm}

The Weisfeiler-Leman algorithm, originally introduced by Weisfeiler and Leman in its $2$-di\-men\-sio\-nal version \cite{WeisfeilerL68}, forms one of the most fundamental subroutines in the context of isomorphism testing.

Let~$\chi_1,\chi_2\colon V^k \rightarrow C$ be colorings of the~$k$-tuples of vertices of~$G$, where~$C$ is a finite set of colors. 
We say $\chi_1$ \emph{refines} $\chi_2$, denoted $\chi_1 \preceq \chi_2$, if $\chi_1(\bar v) = \chi_1(\bar w)$ implies $\chi_2(\bar v) = \chi_2(\bar w)$ for all $\bar v,\bar w \in V^{k}$.
The colorings $\chi_1$ and $\chi_2$ are \emph{equivalent}, denoted $\chi_1 \equiv \chi_2$,  if $\chi_1 \preceq \chi_2$ and $\chi_2 \preceq \chi_1$.

We describe the \emph{$k$-dimensional Weisfeiler-Leman algorithm} ($k$-WL) for all $k \geq 1$.
For an input graph $G$ let $\WLit{k}{0}{G}\colon (V(G))^{k} \rightarrow C$ be the coloring where each tuple is colored with the isomorphism type of its underlying ordered subgraph.
More precisely, $\WLit{k}{0}{G}(v_1,\dots,v_k) = \WLit{k}{0}{G}(v_1',\dots,v_k')$ if and only if, for all $i,j \in [k]$, it holds that
$v_i = v_j \Leftrightarrow v_i'= v_j'$ and $v_iv_j \in E(G) \Leftrightarrow v_i'v_j' \in E(G)$.
If the graph is equipped with a coloring the initial coloring $\WLit{k}{0}{G}$ also takes the input coloring into account.
More precisely, for a vertex-coloring $\chi_V$, it additionally holds that $\chi_V(v_i) = \chi_V(v_i')$ for all $i \in [k]$.
And for an arc-coloring $\chi_E$, it is the case that $\chi_E(v_i,v_j) = \chi_E(v_i',v_j')$ for all $i,j \in [k]$ such that $v_iv_j \in E(G)$.

We then recursively define the coloring $\WLit{k}{i}{G}$ obtained after $i$ rounds of the algorithm.
For $k \geq 2$ and $\bar v = (v_1,\dots,v_k) \in (V(G))^k$ let

\[\WLit{k}{i+1}{G}(\bar v) \coloneqq \Big(\WLit{k}{i}{G}(\bar v), \CM_i(\bar v)\Big)\]
where
\[\CM_i(\bar v) \coloneqq \Big\{\!\!\Big\{\big(\WLit{k}{i}{G}(\bar v[w/1]),\dots,\WLit{k}{i}{G}(\bar v[w/k])\big) \;\Big\vert\; w \in V(G) \Big\}\!\!\Big\}\]
and $\bar v[w/i] \coloneqq (v_1,\dots,v_{i-1},w,v_{i+1},\dots,v_k)$ is the tuple obtained from $\bar v$ by replacing the $i$-th entry by $w$ (and $\{\!\{\dots\}\!\}$ denotes a multiset).
For $k = 1$ the definition is similar, but we only iterate over neighbors of $v$, i.e., $\WLit{1}{i+1}{G}(v) \coloneqq \Big(\WLit{1}{i}{G}(v), \CM_i(v)\Big)$
where
\[\CM_i(v) \coloneqq \Big\{\!\!\Big\{\WLit{1}{i}{G}(w) \;\Big\vert\; w \in N_G(v) \Big\}\!\!\Big\}.\]
By definition, $\WLit{k}{i+1}{G} \preceq \WLit{k}{i}{G}$ for all $i \geq 0$.
Hence, there is a minimal~$i_\infty$ such that $\WLit{k}{i_{\infty}}{G} \equiv \WLit{k}{i_{\infty}+1}{G}$ and for this $i_\infty$ the coloring $\WL{k}{G} \coloneqq \WLit{k}{i_\infty}{G}$ is the \emph{$k$-stable} coloring of $G$.

The $k$-dimensional Weisfeiler-Leman algorithm takes as input a (vertex- or arc-)colored graph $G$ and returns (a coloring that is equivalent to) $\WL{k}{G}$.
This can be implemented in time $\CO(n^{k+1}\log n)$ (see \cite{ImmermanL90}).

\subsection{Group Theory}
\label{sec:group-theory}

We introduce the group-theoretic notions required in this work.
For a general background on group theory we refer to \cite{Rotman99}, whereas background on permutation groups can be found in \cite{DixonM96}.

\paragraph{Permutation Groups}

A \emph{permutation group} acting on a set $\Omega$ is a subgroup $\Gamma \leq \Sym(\Omega)$ of the symmetric group.
The size of the permutation domain $\Omega$ is called the \emph{degree} of $\Gamma$.
If $\Omega = [n]$, then we also write $S_n$ instead of $\Sym(\Omega)$.
For $\gamma \in \Gamma$ and $\alpha \in \Omega$ we denote by $\alpha^{\gamma}$ the image of $\alpha$ under the permutation $\gamma$.
For $A \subseteq \Gamma$, we denote by $\Gamma_{(A)} \coloneqq \{\gamma \in \Gamma \mid \forall \alpha \in A\colon \alpha^\gamma = \alpha\}$ is pointwise stabilizer of $A$ in $\Gamma$.
For $A \subseteq \Omega$ and $\gamma \in \Gamma$ let $A^{\gamma} \coloneqq \{\alpha^{\gamma} \mid \alpha \in A\}$.
The set $A$ is \emph{$\Gamma$-invariant} if $A^{\gamma} = A$ for all $\gamma \in \Gamma$.
For a partition $\CP$ of $\Omega$ let $\CP^\gamma \coloneqq \{A^\gamma \mid A \in \CP\}$.
Observe that $\CP^\gamma$ is again a partition of $\Omega$.
We say $\CP$ is \emph{$\Gamma$-invariant} if $\CP^{\gamma} = \CP$ for all $\gamma \in \Gamma$.

For $A \subseteq \Omega$ and a bijection $\theta\colon \Omega \rightarrow \Omega'$ we denote by $\theta[A]$ the restriction of $\theta$ to the domain $A$.
For a $\Gamma$-invariant set $A \subseteq \Omega$, we denote by $\Gamma[A] \coloneqq \{\gamma[A] \mid \gamma \in \Gamma\}$ the induced action of $\Gamma$ on $A$, i.e., the group obtained from $\Gamma$ by restricting all permutations to $A$.
More generally, for every set $\Lambda$ of bijections with domain $\Omega$, we denote by $\Lambda[A] \coloneqq \{\theta[A] \mid \theta \in \Lambda\}$.
Similarly, for a partition $\CP$ of $\Omega$, we denote by $\theta[\CP]\colon \CP \rightarrow \CP'$ the mapping defined via $\theta(A) \coloneqq \{\theta(\alpha) \mid \alpha \in A\}$ for all $A \in \CP$.
As before, $\Lambda[\CP] \coloneqq \{\theta[\CP] \mid \theta \in \Lambda\}$.

\paragraph{Algorithms for Permutation Groups}

Next, let us review some basic facts about algorithms for permutation groups.
More details can be found in \cite{Seress03}.

In order to perform computational tasks for permutation groups efficiently the groups are represented by generating sets of small size.
Indeed, most algorithms are based on so-called strong generating sets,
which can be chosen of size quadratic in the size of the permutation domain of the group and can be computed in polynomial time given an arbitrary generating set (see, e.g., \cite{Seress03}).

\begin{theorem}[cf.\ \cite{Seress03}] 
 \label{thm:permutation-group-library}
 Let $\Gamma \leq \Sym(\Omega)$ and let $S$ be a generating set for $\Gamma$.
 Then the following tasks can be performed in time polynomial in $|\Omega|$ and $|S|$:
 \begin{enumerate}
  \item compute the order of $\Gamma$,
  \item given $\gamma \in \Sym(\Omega)$, test whether $\gamma \in \Gamma$,
  \item compute the orbits of $\Gamma$, and
  \item given $A \subseteq \Omega$, compute a generating set for $\Gamma_{(A)}$.
 \end{enumerate}
\end{theorem}

\paragraph{Groups with Restricted Composition Factors}

In this work, we shall be interested in a particular subclass of permutation groups, namely groups with restricted composition factors.
Let $\Gamma$ be a group.
A \emph{subnormal series} is a sequence of subgroups $\Gamma = \Gamma_0 \geq \Gamma_1 \geq \dots \geq \Gamma_k = \{\id\}$ such that $\Gamma_i$ is a normal subgroup of $\Gamma_{i-1}$ for all $i \in [k]$.
The length of the series is $k$ and the groups $\Gamma_{i-1} / \Gamma_{i}$ are the factor groups of the series, $i \in [k]$.
A \emph{composition series} is a strictly decreasing subnormal series of maximal length. 
For every finite group $\Gamma$ all composition series have the same family (considered as a multiset) of factor groups (cf.\ \cite{Rotman99}).
A \emph{composition factor} of a finite group $\Gamma$ is a factor group of a composition series of $\Gamma$.

\begin{definition}
 For $d \geq 2$ let $\mgamma_d$ denote the class of all groups $\Gamma$ for which every composition factor of $\Gamma$ is isomorphic to a subgroup of $S_d$.
 A group $\Gamma$ is a \emph{$\mgamma_d$-group} if it is contained in the class $\mgamma_d$.
\end{definition}

Let us point out the fact that there are two similar classes of groups usually referred by $\Gamma_d$ in the literature.
The first is the class denoted by $\mgamma_d$ here originally introduced by Luks \cite{Luks82}, while the second one, for example used in \cite{BabaiCP82}, in particular allows composition factors that are simple groups of Lie type of bounded dimension.

\paragraph{Group-Theoretic Tools for Isomorphism Testing}

In this work, the key group-theoretic subroutine is an isomorphism test for hypergraphs where the input group is a $\mgamma_d$-group.
Two hypergraphs $\CH_1 = (V_1,\CE_1)$ and $\CH_2 = (V_2,\CE_2)$ are isomorphic if there is a bijection $\varphi\colon V_1 \rightarrow V_2$ such that $E \in \CE_1$ if and only if $E^{\varphi} \in \CE_2$ for all $E \in 2^{V_1}$
(where $E^\varphi \coloneqq \{\varphi(v) \mid v \in E\}$ and $2^{V_1}$ denotes the power set of $V_1$).
We write $\varphi\colon \CH_1 \cong \CH_2$ to denote that $\varphi$ is an isomorphism from $\CH_1$ to $\CH_2$.
Consistent with previous notation, we denote by $\Iso(\CH_1,\CH_2)$ the set of isomorphisms from $\CH_1$ to $\CH_2$.
More generally, for $\Gamma \leq \Sym(V_1)$ and a bijection $\theta\colon V_1 \rightarrow V_2$, we define
\[\Iso_{\Gamma\theta}(\CH_1,\CH_2) \coloneqq \{\varphi \in \Gamma\theta \mid \varphi\colon \CH_1 \cong \CH_2\}.\]
In this work, we define the Hypergraph Isomorphism Problem to take as input two hypergraphs $\CH_1 = (V_1,\CE_1)$ and $\CH_2 = (V_2,\CE_2)$, a group $\Gamma \leq \Sym(V_1)$ and a bijection $\theta\colon V_1 \rightarrow V_2$, and the goal is to compute a representation of $\Iso_{\Gamma\theta}(\CH_1,\CH_2)$.
The following algorithm forms a crucial subroutine.

\begin{theorem}[Miller \cite{Miller83b}]
 \label{thm:hypergraph-isomorphism-gamma-d}
 Let $\CH_1 = (V_1,\CE_1)$ and $\CH_2 = (V_2,\CE_2)$ be two hypergraphs and let $\Gamma \leq \Sym(V_1)$ be a $\mgamma_d$-group and $\theta\colon V_1 \rightarrow V_2$ a bijection.
 Then $\Iso_{\Gamma\theta}(\CH_1,\CH_2)$ can be computed in time $(n+m)^{\CO(d)}$ where $n \coloneqq |V_1|$ and $m \coloneqq |\CE_1|$.
\end{theorem}

\begin{theorem}[Neuen \cite{Neuen22b}]
 \label{thm:hypergraph-isomorphism-gamma-d-fast}
 Let $\CH_1 = (V_1,\CE_1)$ and $\CH_2 = (V_2,\CE_2)$ be two hypergraphs and let $\Gamma \leq \Sym(V_1)$ be a $\mgamma_d$-group and $\theta\colon V_1 \rightarrow V_2$ a bijection.
 Then $\Iso_{\Gamma\theta}(\CH_1,\CH_2)$ can be computed in time $(n+m)^{\CO((\log d)^c)}$ for some constant $c$ where $n \coloneqq |V_1|$ and  $m \coloneqq |\CE_1|$.
\end{theorem}

Observe that both algorithms given by the two theorems tackle the same problem.
The second algorithm is asymptotically much faster, but it is also much more complicated and the constant factors in the exponent of the running time are likely to be much higher.
Since this paper only applies either theorem for $d=2$, it seems to be preferable to use the first algorithm.
Indeed, the first result is a simple extension of Luks's well-known isomorphism test for bounded-degree graphs \cite{Luks82}, and thus the underlying algorithm is fairly simple.
For all these reasons, we mostly build on Theorem \ref{thm:hypergraph-isomorphism-gamma-d}.
However, for future applications of the techniques presented in this work, it might be necessary to build on Theorem \ref{thm:hypergraph-isomorphism-gamma-d-fast} to benefit from the improved run time bound.
For this reason, we shall provide variants of our results building on Theorem \ref{thm:hypergraph-isomorphism-gamma-d-fast} wherever appropriate.

For the sake of completeness, I remark that there is another algorithm tackling the problem covered by both theorems running in time $n^{\CO(d)}m^{\CO(1)}$ \cite{SchweitzerW19}.
However, this variant is not beneficial to us since $m = \CO(n)$ for all applications considered in this paper.

\section{Allowing Weisfeiler and Leman to Split Small Color Classes}
\label{sec:t-k-wl}

In this section, we introduce the concept of \emph{$(t,k)$-WL-bounded graphs} and provide a polynomial-time isomorphism test for such graphs for all constant values of $t$ and $k$.
The final fpt isomorphism test for graphs excluding $K_{3,h}$ as a minor builds on this subroutine for $t = k = 2$.

The concept of $(t,k)$-WL-bounded graphs is a natural extension of $t$-CR-bounded graphs which were already introduced by Ponomarenko in the late 1980's \cite{Ponomarenko89} and which were recently rediscovered in \cite{Neuen22b,GroheNW23,Neuen22}.
Intuitively speaking, a graph $G$ is \emph{$t$-CR-bounded}, $t \in \NN$, if an initially uniform vertex-coloring $\chi$ (i.e., all vertices receive the same color) can be turned into the discrete coloring (i.e., each vertex has its own color) by repeatedly
\begin{itemize}
 \item performing the Color Refinement algorithm (expressed by the letters `CR'), and
 \item taking a color class $[v]_\chi \coloneqq \{w \in V(G) \mid \chi(w) = \chi(v)\}$ of size $|[v]_\chi| \leq t$ and assigning each vertex from the class its own color.
\end{itemize}

A very natural extension of this idea to replace the Color Refinement algorithm by the Weisfeiler-Leman algorithm for some fixed dimension $k$.
This leads us to the notion of \emph{$(t,k)$-WL-bounded graphs} (the letters `CR' are replaced by `$k$-WL').
In particular, $(t,1)$-WL-bounded graphs are exactly the $t$-CR-bounded graphs.
Maybe surprisingly, it seems that this simple extension has not been considered so far in the literature.

\begin{definition}
 \label{def:t-cr-bounded}
 A vertex- and arc-colored graph $G = (V,E,\chi_V,\chi_E)$ is \emph{$(t,k)$-WL-bounded}
 if the sequence $(\chi_i)_{i \geq 0}$ reaches a discrete coloring where $\chi_0 \coloneqq \chi_V$,
 \[\chi_{2i+1}(v) \coloneqq \WL{k}{V,E,\chi_{2i},\chi_E}(v,\dots,v)\]
 and
 \[\chi_{2i+2}(v) \coloneqq \begin{cases}
                     (v,1)              & \text{if } |[v]_{\chi_{2i+1}}| \leq t\\
                     (\chi_{2i+1}(v),0) & \text{otherwise}
                    \end{cases}\]
 for all $i \geq 0$.
 
 Also, for the minimal $i_\infty \geq 0$ such that $\chi_{i_\infty} \equiv \chi_{i_\infty+1}$, we refer to $\chi_{i_\infty}$ as the \emph{$(t,k)$-WL-stable} coloring of $G$ and denote it by $\tWL{t}{k}{G}$.
\end{definition}

At this point, the reader may wonder why $(\chi_i)_{i \geq 0}$ is chosen as a sequence of vertex-colorings and not a sequence of colorings of $k$-tuples of vertices (since $k$-WL also colors $k$-tuples of vertices).
While such a variant certainly makes sense, it still leads to the same class of graphs.
Let $G$ be a graph and let $\chi \coloneqq \WL{k}{G}$.
The main insight is that, if there is some color $c \in \im(\chi)$ for which $|\chi^{-1}(c)| \leq t$, then there is also a color $c' \in \im(\chi)$ for which $|\chi^{-1}(c')| \leq t$ and $\chi^{-1}(c') \subseteq \{(v,\dots,v) \mid v \in V(G)\}$.
In other words, one can not achieve any additional splitting of color classes by also considering non-diagonal color classes.

We also need to extend several notions related to $t$-CR-bounded graphs.
Let $G$ be a graph and let $X \subseteq V(G)$ be a set of vertices.
Let $\chi_X^*\colon V(G) \rightarrow C$ be the vertex-coloring obtained from individualizing all vertices in the set $X$, i.e., $\chi_X^*(v) \coloneqq (v,1)$ for $v \in X$ and $\chi_X^*(v) \coloneqq (0,0)$ for $v \in V(G) \setminus X$.
Let $\chi \coloneqq \tWL{t}{k}{G,\chi_X^*}$ denote the $(t,k)$-WL-stable coloring with respect to the input graph $(G,\chi_X^*)$.
We define the \emph{$(t,k)$-closure} of the set $X$ (with respect to $G$) to be the set
\[\cl_{t,k}^G(X) \coloneqq \left\{v \in V(G) \mid |[v]_{\chi}| = 1\right\}.\]
It is not difficult to verify that $\cl_{t,k}^G$ is a closure operator.
Indeed, $X \subseteq \cl_{t,k}^G(X)$ for all $X \subseteq V(G)$ by definition.
If $X \subseteq Y$ then $\chi_Y^* \preceq \chi_X^*$ which implies that $\tWL{t}{k}{G,\chi_Y^*} \preceq \tWL{t}{k}{G,\chi_X^*}$.
So $\cl_{t,k}^G(X) \subseteq \cl_{t,k}^G(Y)$.
Finally, $\cl_{t,k}^G(X) = \cl_{t,k}^G(\cl_{t,k}^G(X))$ since applying the $(t,k)$-WL procedure to a coloring $\chi_X^*$ for a second time does not lead to a finer coloring.

For $v_1,\dots,v_\ell \in V(G)$ we use $\cl_{t,k}^G(v_1,\dots,v_\ell)$ as a shorthand for $\cl_{t,k}^G(\{v_1,\dots,v_\ell\})$.
If the input graph is equipped with a vertex- or arc-coloring, all definitions are extended in the natural way. 

For the rest of this section, we concern ourselves with designing a polynomial-time isomorphism test for $(t,k)$-WL-bounded graphs.
Actually, we shall prove a slightly stronger result which turns out to be useful later on.
The main idea for the algorithm is to build a reduction to the isomorphism problem for $(t,1)$-WL-bounded graphs for which such results are already known \cite{Ponomarenko89,Neuen22b}.
Indeed, isomorphism of $(t,1)$-WL-bounded graphs can be reduced to the Hypergraph Isomorphism Problem for $\mgamma_t$-groups.
Since one may be interested in using different subroutines for solving the Hypergraph Isomorphism Problem for $\mgamma_t$-groups (see the discussion at the end of Section \ref{sec:group-theory}), most results are stated via an oracle for the Hypergraph Isomorphism Problem on $\mgamma_t$-groups.

\begin{theorem}[\mbox{\cite[Lemma 5.5]{Neuen22b}}]
 \label{thm:isomorphism-t-cr-bounded}
 Let $G_1$ and $G_2$ be two vertex- and arc-colored graphs that are $(t,1)$-WL-bounded.
 Then there is an algorithm computing $\Iso(G_1,G_2)$ in polynomial time using oracle access to the Hypergraph Isomorphism Problem for $\mgamma_t$-groups.
 Moreover, $\Aut(G_1) \in \mgamma_t$.
\end{theorem}

\begin{theorem}
 \label{thm:bounding-group-t-k-wl}
 Let $G_1$ and $G_2$ be two vertex- and arc-colored graphs and let $\chi_i \coloneqq \tWL{t}{k}{G_i}$.
 Also let $\CP_i = \{[v]_{\chi_i} \mid v \in V(G_i)\}$ be the partition into color classes of $\chi_i$.
 Then $\CP_1^\varphi = \CP_2$ for all $\varphi \in \Iso(G_1,G_2)$.
 
 Moreover, using oracle access to the Hypergraph Isomorphism Problem for $\mgamma_t$-groups,
 in time $n^{\CO(k)}$ one can compute a $\mgamma_t$-group $\Gamma \leq \Sym(\CP_1)$ and a bijection $\theta \colon \CP_1 \rightarrow \CP_2$ such that
 \[\Iso(G_1,G_2)[\CP_1] \subseteq \Gamma\theta.\]
 In particular, $\Aut(G_1)[\CP_1] \in \mgamma_t$.
\end{theorem}

\begin{proof}
 Let $G_1,G_2$ be two vertex- and arc-colored graphs.
 Clearly, $\CP_1^\varphi = \CP_2$ for every $\varphi \in \Iso(G_1,G_2)$ since the partition $\CP_i$ is defined in an isomorphism-invariant manner.
 Since the proof of the second part is rather technical, let us first provide a brief intuitive overview.
 The basic idea is to construct $t$-CR-bounded graphs $H_1,H_2$ that model the computation of the colorings $\chi_1$ and $\chi_2$.
 In particular, $\CP_i$ can be viewed as a subset of the vertex set of $H_i$, and restricting all isomorphisms between $H_1$ and $H_2$ to the vertex subsets associated with $\CP_1$ and $\CP_2$ gives the desired output $\Gamma\theta$.

 To build the graph $H_i$ ($i \in \{1,2\}$), we follow the iterative process of computing the coloring $\chi_i$.
 Since $k$-WL computes a coloring on $k$-tuples, we first extend the coloring to color all $k$-tuples of vertices.
 Additionally, instead of applying $k$-WL in a single step, we unfold its definition and only apply a single round of $k$-WL at a time.
 So overall, we obtain a sequence of colorings where in a single step we either split (diagonal) color classes of size at most $t$, or we perform a single round of $k$-WL.

 For each coloring in this sequence, we add a \emph{layer} to $H_i$ that models the corresponding refinement step (see Figures \ref{fig:t-k-wl-reduction-split} and \ref{fig:t-k-wl-reduction-refine}).
 This is done in such a way that all refinement steps are also performed within $t$-CR which eventually means that $H_1,H_2$ are $t$-CR-bounded and we can use Theorem \ref{thm:isomorphism-t-cr-bounded} to obtain the desired output $\Gamma\theta$.

 Let us now turn to the formal proof.
 The algorithm first inductively computes the following sequence of colorings $\chi_{j,r}^i \colon (V(G_i))^k \rightarrow C$ for all $j \in [j_\infty]$ and $r \in [r_\infty(j)]$ (where $C$ is a suitable set of colors).
 We set
 \[\chi_{1,1}^i \coloneqq \WL{k}{G_i}\]
 and set $r_\infty(1) \coloneqq 1$.
 Now, suppose that $j > 1$.
 If $j$ is even, then we define
 \[\chi_{j,1}^i(v_1,\dots,v_k) \coloneqq \begin{cases}
                                          (v_1,1)              & \text{if } v_1 = v_2 = \dots = v_k \text{ and }\\
                                                               &|C(v;\chi_{j-1,r_\infty(j-1)}^i)| \leq t\\
                                          (\chi_{j-1,r_\infty(j-1)}^i(v_1,\dots,v_k),0) & \text{otherwise}
                                         \end{cases}\]
 where
 \[C(v;\chi_{j-1,r_\infty(j-1)}^i) \coloneqq \{w \in V(G_i) \mid \chi_{j-1,r_\infty(j-1)}^i(w,\dots,w) = \chi_{j-1,r_\infty(j-1)}^i(v,\dots,v)\}.\]
 Also, we set $r_\infty(j) \coloneqq 1$.
 If $\chi_{j,1}^i \equiv \chi_{j-1,r_\infty(j-1)}^i$ the algorithm terminates with $j_\infty \coloneqq j-1$.
 
 Otherwise $j$ is odd. For notational convenience, we set $\chi_{j,0}^i \coloneqq \chi_{j-1,r_\infty(j-1)}^i$.
 For $r \geq 1$ and $\bar v = (v_1,\dots,v_k) \in (V(G_i))^k$ let
 \[\chi_{j,r+1}^i(\bar v) \coloneqq \Big(\chi_{j,r}^i(\bar v), \CM_r(\bar v)\Big)\]
 where
 \[\CM_r(\bar v) \coloneqq \Big\{\!\!\Big\{\big(\chi_{j,r}^i(\bar v[w/1]),\dots,\chi_{j,r}^i(\bar v[w/k])\big) \;\Big\vert\; w \in V(G_i) \Big\}\!\!\Big\}\]
 and $\bar v[w/i] \coloneqq (v_1,\dots,v_{i-1},w,v_{i+1},\dots,v_k)$ is the tuple obtained from replacing the $i$-th entry by $w$.
 We set $r_\infty(j) \coloneqq r$ for the minimal $r \geq 0$ such that $\chi_{j,r}^i \equiv \chi_{j,r+1}^i$.
 If $r = 0$ then the algorithm terminates and $j_\infty \coloneqq j-1$.
 This completes the description of the colorings.
 
 The definition of the sequence of colorings clearly follows the definition of the $(t,k)$-WL-stable coloring where the single refinement steps of the WL-algorithm are given explicitly.
 Hence,
 \begin{equation}
  \tWL{t}{k}{G_i} \equiv \chi_{j_\infty,r_\infty(j_\infty)}^i.
 \end{equation}
 
 Now, the idea is to define isomorphism-invariant vertex- and arc-colored graphs $(H_i,\lambda_V^i,\lambda_E^i)$ such that $\im(\chi_{j,r}^i)$ corresponds to a subset of $V(H_i)$ for all $j \in [j_\infty]$ and $r \in [r_\infty(j)]$, and $(H_i,\lambda_V^i,\lambda_E^i)$ is $t$-CR-bounded.
 Towards this end, we partition the vertex set of $H_i$ into \emph{layers} each of which corresponds to one of the colorings.
 To be more precise, we define
 \[V(H_i) \coloneqq \bigcup_{j \in [j_\infty]} \bigcup_{r \in [r_\infty(j)]} V_{j,r}^i\]
 for some suitable sets $V_{j,r}^i$ to be described below.
 Let $j \in [j_\infty]$ and $r \in [r_\infty(j)]$.
 We define
 \[W_{j,r}^i \coloneqq \{(i,j,r,c) \mid c \in \im(\chi_{j,r}^i)\} \subseteq V_{j,r}^i\]
 and set
 \[\lambda_V^i(i,j,r,c) \coloneqq \begin{cases}
                                   (j,r,{\sf col})   &\text{if } (j,r) \neq (1,1)\\
                                   (j,r,{\sf col},c) &\text{otherwise}
                                  \end{cases}.\]
 Here, ${\sf col}$ is a special value indicating that $(i,j,r,c)$ is a vertex that corresponds to a color of $\chi_{j,r}^i$.
 (In connecting the layers, we will need to add further auxiliary vertices to build certain gadgets. These vertices will be colored using a special value ${\sf aux}$ in the third component.)
 
 We now connect the layers in such a way that $(H_i,\lambda_V^i,\lambda_E^i)$ is $t$-CR-bounded.
 Towards this end, we build on the inductive definition of the colorings.
 Let
 \[X_i \coloneqq \{v \in V(H_i) \mid |[v]_{\tWL{t}{1}{H_i}}| = 1\}\]
 denote the set of vertices appearing in a singleton color class of $\tWL{t}{1}{H_i}$.
 Throughout the construction, we maintain the property that $V_{j,r}^i \subseteq X_i$ for all layers covered so far.
 
 For the first layer $j=1$, $r=1$ there is nothing more to be done.
 By the definition of $\lambda_V^i$ all vertices $V_{1,1}^i \coloneqq W_{1,1}^i$ already form singleton color classes.
 
 \begin{figure}
  \centering
  \begin{tikzpicture}
   \node at (-3,2) {$W_{j-1,r_\infty(j-1)}^i$};
   \node at (-3,0) {$W_{j,1}^i$};

   \node at (-1,2) {$\dots$};
   \node[vertex, fill = mBlue, label = {[align=left]above:{\scriptsize $(v_1,v_2)$}\\{\scriptsize $(v_2,v_1)$}}] (b1) at (0,2) {};
   \node[vertex, fill = mBlue, label = {[align=left]above:{\scriptsize $(v_1,v_1)$}\\{\scriptsize $(v_2,v_2)$}}] (b2) at (2.25,2) {};
   \node[vertex, fill = mBlue, label = {[align=left]above:{\scriptsize $(v_3,v_3)$}\\{\scriptsize $(v_4,v_4)$}\\{\scriptsize $(v_5,v_5)$}}] (b3) at (4.5,2) {};
   \node[vertex, fill = mBlue, label = {[align=left]above:{\scriptsize $(v_3,v_4)$}\\{\scriptsize $(v_4,v_5)$}\\{\scriptsize $(v_5,v_3)$}}] (b4) at (6,2) {};
   \node[vertex, fill = mBlue, label = {[align=left]above:{\scriptsize $(v_4,v_3)$}\\{\scriptsize $(v_5,v_4)$}\\{\scriptsize $(v_3,v_5)$}}] (b5) at (7.5,2) {};
   \node at (8.5,2) {$\dots$};

   \node at (-1,0) {$\dots$};
   \node[vertex, fill = mRed, label = {[align=left]below:{\scriptsize $(v_1,v_2)$}\\{\scriptsize $(v_2,v_1)$}}] (r1) at (0,0) {};
   \node[vertex, fill = mRed, label = {[align=left]below:{\scriptsize $(v_1,v_1)$}}] (r21) at (1.5,0) {};
   \node[vertex, fill = mRed, label = {[align=left]below:{\scriptsize $(v_2,v_2)$}}] (r22) at (3,0) {};
   \node[vertex, fill = mRed, label = {[align=left]below:{\scriptsize $(v_3,v_3)$}\\{\scriptsize $(v_4,v_4)$}\\{\scriptsize $(v_5,v_5)$}}] (r3) at (4.5,0) {};
   \node[vertex, fill = mRed, label = {[align=left]below:{\scriptsize $(v_3,v_4)$}\\{\scriptsize $(v_4,v_5)$}\\{\scriptsize $(v_5,v_3)$}}] (r4) at (6,0) {};
   \node[vertex, fill = mRed, label = {[align=left]below:{\scriptsize $(v_4,v_3)$}\\{\scriptsize $(v_5,v_4)$}\\{\scriptsize $(v_3,v_5)$}}] (r5) at (7.5,0) {};
   \node at (8.5,0) {$\dots$};

   \foreach \v/\w in {b1/r1,b2/r21,b2/r22,b3/r3,b4/r4,b5/r5}{
    \draw[thick] (\v) edge (\w);
   }
  \end{tikzpicture}
  \caption{Visualization of the construction of layer $(j,1)$ ($j$ even) of the graph $H_i$ in the proof of Theorem \ref{thm:bounding-group-t-k-wl}.
  Every vertex in $W_{j,1}^i$ is associated with a color $c$ in the image of $\chi_{j,1}^i$ and the corresponding color class $(\chi_{j,1}^i)^{-1}(c)$ shown next to the vertex.
  All color classes of size at most $t$ (the figure shows $t = 2$) containing only diagonal entries are split into singletons.
  }
  \label{fig:t-k-wl-reduction-split}
 \end{figure}

 So suppose $j > 1$.
 First suppose $j$ is even (see Figure \ref{fig:t-k-wl-reduction-split}).
 In this case $r_\infty(j) = 1$.
 We set $V_{j,1}^i \coloneqq W_{j,1}^i$
 For each $c \in \im(\chi_{j,1}^i)$ we add an edge $\{(i,j,1,c),(i,j-1,r_\infty(j-1),c')\}$ where $c' \in \im(\chi_{j-1,r_\infty(j-1)}^i)$ is the unique color for which
 \[\big(\chi_{j,1}^i\big)^{-1}(c) \subseteq \big(\chi_{j-1,r_\infty(j-1)}^i\big)^{-1}(c').\]
 Also, we set
 \[\lambda_E^i((i,j,1,c)(i,j-1,r_\infty(j-1),c')) \coloneqq 0.\]
 By definition of the coloring $\chi_{j,1}^i$, every vertex $v \in V_{j-1,r_\infty(j-1)}^i$ has at most $t$ many neighbors in $V_{j,1}^i$.
 Since $V_{j-1,r_\infty(j-1)}^i \subseteq X_i$ by induction, this implies that $V_{j,1}^i \subseteq X_i$.
 
 Next, suppose $j$ is odd (see Figure \ref{fig:t-k-wl-reduction-refine}).
 For notational convenience, we set $V_{j,0}^i \coloneqq V_{j-1,r_\infty(j-1)}^i$ as well as $(i,j,0,c) \coloneqq (i,j-1,r_\infty(j-1),c)$ for all vertices $(i,j-1,r_\infty(j-1),c) \in W_{j-1,r_\infty(j-1)}^i$.
 Fix some $r \in [r_\infty(j)]$.
 Suppose $\im(\chi_{j,r}^i) = \{c_1,\dots,c_\ell\}$.
 For every $p \in [\ell]$ we have that $c_p = (c_p',\CM_p)$ where $c_p' \in \im(\chi_{j,r-1}^i)$, and $\CM_p$ is a multiset over elements from $\big(\im(\chi_{j,r-1}^i)\big)^k$.
 Let $\CM \coloneqq \bigcup_{p \in [\ell]} \CM_p$.
 For each $\bar c = (\bar c_1,\dots,\bar c_k) \in \CM$ we introduce a vertex $(i,j,r,\bar c)$.
 We set $\lambda_V^i(i,j,r,\bar c) \coloneqq (j,r,{\sf aux})$ and connect $(i,j,r,\bar c)$ to all vertices $(i,j,r-1,\bar c_q)$ for all $q \in [k]$.
 We set
 \[\lambda_E^i((i,j,r,\bar c)(i,j,r-1,\bar c_q)) \coloneqq \{q' \in [k] \mid \bar c_q = \bar c_{q'}\}\]
 for all $q \in [k]$.
 Next, we connect $(i,j,r,c_p)$ to $(i,j,r-1,c_p')$ as well as to $(i,j,r,\bar c)$ for all $\bar c \in \CM_p$ and $p \in [\ell]$.
 We set
 \[\lambda_E^i((i,j,r,c_p)(i,j,r-1,c_p')) \coloneqq 0\]
 and $\lambda_E^i((i,j,r,c_p)(i,j,r,\bar c))$ to the multiplicity of $\bar c$ in the multiset $\CM_p$.
 
 \begin{figure}
  \centering
  \begin{tikzpicture}
   \node at (-3,4) {$W_{j,r-1}^i$};
   \node at (-3,0) {$W_{j,r}^i$};

   \node at (-1,4) {$\dots$};
   \node at (-1,2) {$\dots$};
   \node at (-1,0) {$\dots$};

   \node at (7,4) {$\dots$};
   \node at (7.5  ,2) {$\dots$};
   \node at (7,0) {$\dots$};

   \node[vertex, fill = mRed, label = {[align=left]above:{\scriptsize $(v_1,v_2)$}\\{\scriptsize $(v_2,v_1)$}}] (r1) at (0,4) {};
   \node[vertex, fill = mRed, label = {[align=left]above:{\scriptsize $(v_1,v_1)$}}] (r2) at (4,4) {};
   \node[vertex, fill = mRed, label = {[align=left]above:{\scriptsize $(v_2,v_2)$}}] (r3) at (6,4) {};

   \node[vertex, fill = mYellow] (y0) at (0.5,2) {};
   \node[vertex, fill = mYellow] (y1) at (1.5,2) {};
   \node[vertex, fill = mYellow] (y2) at (2.5,2) {};
   \node[vertex, fill = mYellow] (y3) at (3.5,2) {};
   \node[vertex, fill = mYellow] (y4) at (4.5,2) {};
   \node[vertex, fill = mYellow] (y5) at (5.5,2) {};
   \node[vertex, fill = mYellow] (y6) at (6.5,2) {};

   \node[vertex, fill = mTurquoise, label = {[align=left]below:{\scriptsize $(v_1,v_2)$}}] (t11) at (0,0) {};
   \node[vertex, fill = mTurquoise, label = {[align=left]below:{\scriptsize $(v_2,v_1)$}}] (t12) at (2,0) {};
   \node[vertex, fill = mTurquoise, label = {[align=left]below:{\scriptsize $(v_1,v_1)$}}] (t2) at (4,0) {};
   \node[vertex, fill = mTurquoise, label = {[align=left]below:{\scriptsize $(v_2,v_2)$}}] (t3) at (6,0) {};

   \foreach \v/\w in {r1/y0,r1/y1,r1/y2,r1/y3,r1/y4,r2/y1,r2/y2,r2/y5,r3/y3,r3/y4,r3/y6,
                      r1/t11,r1/t12,r2/t2,r3/t3,
                      y1/t11,y4/t11,y2/t12,y3/t12,y0/t2,y5/t2,y0/t3,y6/t3}{
    \draw[thick] (\v) edge (\w);
   }

   \begin{pgfonlayer}{background}
    \foreach \v/\w in {r1/y1,r1/y3,r2/y2,r3/y4,
                       y1/t11,y4/t11,y2/t12,y3/t12,y0/t2,y5/t2,y0/t3,y6/t3}{
     \draw[line width = 4.0pt, mOrange!80] (\v.center) edge (\w.center);
    }

    \foreach \v/\w in {r1/y2,r1/y4,r2/y1,r3/y3}{
     \draw[line width = 4.0pt, mViolet!80] (\v.center) edge (\w.center);
    }
    \draw[line width = 4.0pt, mOrange!80, dash pattern = on 3pt off 3pt, dash phase = 3pt] (r1.center) edge (y0.center);
    \draw[line width = 4.0pt, mViolet!80, dash pattern = on 3pt off 3pt] (r1.center) edge (y0.center);
    \draw[line width = 4.0pt, mOrange!80, dash pattern = on 3pt off 3pt, dash phase = 3pt] (r2.center) edge (y5.center);
    \draw[line width = 4.0pt, mViolet!80, dash pattern = on 3pt off 3pt] (r2.center) edge (y5.center);
    \draw[line width = 4.0pt, mOrange!80, dash pattern = on 3pt off 3pt, dash phase = 3pt] (r3.center) edge (y6.center);
    \draw[line width = 4.0pt, mViolet!80, dash pattern = on 3pt off 3pt] (r3.center) edge (y6.center);
   \end{pgfonlayer}
  \end{tikzpicture}
  \caption{Visualization of the construction of layer $(j,r)$ ($j$ odd) of the graph $H_i$ in the proof of Theorem \ref{thm:bounding-group-t-k-wl}.
  Every vertex in $W_{j,r}^i$ is associated with a color $c$ in the image of $\chi_{j,r}^i$ and the corresponding color class $(\chi_{j,r}^i)^{-1}(c)$ shown next to the vertex.
  The edge colors $1$ and $2$ are depicted in orange and violet, respectively.
  For example, the left-most yellow vertex represents the color pair $(\chi_{j,r-1}^i(v_1,v_2), \chi_{j,r-1}^i(v_2,v_1)) = (\chi_{j,r-1}^i(v_2,v_1), \chi_{j,r-1}^i(v_1,v_2))$ and thus, the incident edge to the vertex representing $\chi_{j,r-1}^i(v_1,v_2) = \chi_{j,r-1}^i(v_2,v_1)$ (the left-most red vertex) has color $\{1,2\}$.
  Also, $\CM_{r-1}(v_1,v_2) = \{\!\{(\chi_{j,r-1}^i(v_1,v_2), \chi_{j,r-1}^i(v_1,v_1)),(\chi_{j,-1}^i(v_2,v_2), \chi_{j,r-1}^i(v_1,v_2)),\dots\}\!\}$ and the two color-pairs in the multiset are represented by the second and fifth yellow vertex from the left.}
  \label{fig:t-k-wl-reduction-refine}
 \end{figure}

 It is easy to see that $V_{j,r}^i \subseteq X_i$ using that $V_{j,r-1}^i \subseteq X_i$ by induction.
 Indeed, once all vertices from $V_{j,r-1}^i$ are individualized, it suffices to apply the Color Refinement algorithm for every vertex in $V_{j,r}^i$ to be assigned a distinct color.
 
 This completes the description of $H_i$.
 In total, $X_i = V(H_i)$ which means that $H_i$ is $(t,1)$-WL-bounded.
 Using Theorem \ref{thm:isomorphism-t-cr-bounded}, it is possible to compute $\Iso(H_1,H_2)$ in time polynomial in the size of $H_1$ and $H_2$ using oracle access to the Hypergraph Isomorphism Problem for $\mgamma_t$-groups.

 So let us analyze the size of the graph $H_i$.
 Since each coloring $\chi_{j,r}^i$ refines the previous coloring, we conclude that there are at most $n^k$ many layers.
 Also, $|W_{j,r}^i| = |\im(\chi_{j,r}^i)| \leq n^k$.
 So it remains to bound the number of auxiliary vertices in a given layer.
 Towards this end, observe that $|\CM| \leq p \cdot n \leq n^{k+1}$.
 Hence, $|V(H_i)| \leq n^k(n^{k+1} + n^k) = n^{\CO(k)}$.
 
 This means it is possible to compute $\Iso(H_1,H_2)$ in time $n^{\CO(k)}$ using oracle access to the Hypergraph Isomorphism Problem for $\mgamma_t$-groups.
 Now, the vertices with color $(j_\infty,r_\infty(j_\infty),{\sf col})$ exactly correspond to the color classes of $\chi_i = \tWL{t}{k}{G_i}$.
 We set
 \[\Gamma\theta \coloneqq \Iso(H_1,H_2)[W_{j_\infty,r_\infty(j_\infty)}^1].\]
 By renaming the elements of the domain in the natural way, we get $\Gamma \leq \Sym(\CP_1)$ and $\theta \colon \CP_1 \rightarrow \CP_2$.
 Since $H_i$ is defined in an isomorphism-invariant manner, it follows that $\Iso(G_1,G_2)[\CP_1] \subseteq \Gamma\theta$.
 Finally, $\Gamma \in \mgamma_t$ by Theorem \ref{thm:isomorphism-t-cr-bounded}.
\end{proof}

\begin{corollary}
 Let $G_1$ and $G_2$ be two $(t,k)$-WL-bounded graphs.
 Then a representation for $\Iso(G_1,G_2)$ can be computed in time $n^{\CO (k \cdot (\log t)^c)}$ for some constant $c$.
\end{corollary}

\begin{proof}
 Let $\CP_1 \coloneqq \{\{v\} \mid v \in V(G_1)\}$ and $\CP_2 = \{\{v\} \mid v \in V(G_2)\}$.
 By Theorem \ref{thm:bounding-group-t-k-wl} and \ref{thm:hypergraph-isomorphism-gamma-d-fast}, there is algorithm computing a $\mgamma_t$-group $\Gamma \leq \Sym(\CP_1)$ and bijection $\theta\colon \CP_1 \rightarrow \CP_2$ such that
 \[\Iso(G_1,G_2)[\CP_1] \subseteq \Gamma\theta\]
 running in time $n^{\CO (k \cdot (\log t)^c)}$.
 Identifying the singleton set $\{v\}$ with the element $v$, one may assume $\Gamma \leq \Sym(V(G_1))$ and $\theta\colon V(G_1) \rightarrow V(G_2)$ such that
 \[\Iso(G_1,G_2) \subseteq \Gamma\theta.\]
 In particular
 \[\Iso(G_1,G_2) = \Iso_{\Gamma\theta}(G_1,G_2).\]
 Interpreting both input graphs as hypergraphs, the latter can be computed in time $n^{\CO ((\log t)^c)}$ using Theorem \ref{thm:hypergraph-isomorphism-gamma-d-fast}.
\end{proof}

\section{Structure Theory and Small Color Classes}
\label{sec:overview}

Having established the necessary tools, we can now turn to the isomorphism test for graphs excluding $K_{3,h}$ as a minor.
We start by giving a high-level overview on the algorithm.
The main idea is to build on the isomorphism test for $(2,2)$-WL-bounded graphs described in the last section.
Let $G_1$ and $G_2$ be two (vertex- and arc-colored) graphs that exclude $K_{3,h}$ as a minor.
Using well-known reduction techniques building on isomorphism-invariant decompositions into triconnected\footnote{A triconnected component is either $3$-connected or a cycle.} components (see, e.g., \cite{HopcroftT72}), we may restrict ourselves to the case where $G_1$ and $G_2$ are $3$-connected.
More precisely, given input graphs $G_1$ and $G_2$, we first compute the isomorphism-invariant decomposition into triconnected components of both input graphs, and then compute isomorphisms between $G_1$ and $G_2$ in a bottom-up fashion along the decomposition relying on an isomorphism test for $3$-connected graphs.

So suppose $G_1$ and $G_2$ are $3$-connected
The algorithm starts by individualizing three vertices.
To be more precise, the algorithm picks three distinct vertices $v_1,v_2,v_3 \in V(G_1)$ and iterates over all choices of potential images $w_1,w_2,w_3 \in V(G_2)$ under some isomorphism between $G_1$ and $G_2$.
Let $X_1 \coloneqq \{v_1,v_2,v_3\}$ and $X_2 \coloneqq \{w_1,w_2,w_3\}$.
Also, let $D_i \coloneqq \cl_{2,2}^{G_i}(X_i)$ denote the $(2,2)$-closure of $X_i$, $i \in \{1,2\}$.
Observe that $D_i$ is defined in an isomorphism-invariant manner given the initial choice of $X_i$.
Building on Theorems \ref{thm:bounding-group-t-k-wl} and \ref{thm:hypergraph-isomorphism-gamma-d} it can be checked in polynomial time whether $G_1$ and $G_2$ are isomorphic restricted to the sets $D_1$ and $D_2$.

Now, the central idea is to follow a decomposition strategy.
Let $Z_1^i,\dots,Z_\ell^i$ denote the vertex sets of the connected components of $G_i - D_i$, and let $S_j^i \coloneqq N_{G_i}(Z_j^i)$ for $j \in [\ell]$ and $i \in \{1,2\}$. 
We recursively compute isomorphisms between all pairs of graphs $G_i[Z_j^i \cup S_j^i]$ for all $j \in [\ell]$ and $i \in \{1,2\}$.
To be able to determine whether all these partial isomorphisms can be combined into a global isomorphism, the crucial insight is that $|S_j^i| < h$ for all $j \in [\ell]$ and $i \in \{1,2\}$.

\begin{lemma}
 \label{la:small-separator}
 Let $G$ be a graph that excludes $K_{3,h}$ as a minor.
 Also let $X \subseteq V(G)$ and define $D \coloneqq \cl_{2,2}^G(X)$.
 Let $Z$ be a connected component of $G - D$.
 Then $|N_G(Z)| < h$.
\end{lemma}

Indeed, this lemma forms one of the main technical contributions of the paper.
I remark that similar statements are exploited in \cite{GroheNW23,Neuen22b,Neuen22} eventually leading to an isomorphism test running in time $n^{\CO((\log h)^c)}$ for all graphs excluding $K_h$ as a topological subgraph.
However, all these variants require the $(t,k)$-closure to be taken for non-constant values of $t$ (i.e., $t = \Omega(h)$).
For the design of an fpt-algorithm, this is infeasible since we can only afford to apply Theorems \ref{thm:bounding-group-t-k-wl} and \ref{thm:hypergraph-isomorphism-gamma-d} for constant values of $t$ and $k$ (since the set $D_i$ might be the entire vertex set of $G_i$).

The lemma above implies that the interplay between $D_i$ and $V(G_i) \setminus D_i$ is simple which allows for a dynamic programming approach.
To be more precise, we can recursively list all elements of the set $\Iso((G_i[Z_j^i \cup S_j^i],S_j^i),(G_{i'}[Z_{j'}^{i'} \cup S_{j'}^{i'}],S_{j'}^{i'}))[S_j^i]$ for all $j,j' \in [\ell]$ and $i,i' \in \{1,2\}$ (i.e., we list all bijections $\sigma\colon S_j^i \rightarrow S_{j'}^{i'}$ that can be extended to an isomorphism between the corresponding subgraphs).
To incorporate this information, we extend the graph $G_i[D_i]$ by simple gadgets obtaining graphs $H_i$ that are $(2,2)$-WL-bounded and such that $G_1 \cong G_2$ if and only if $H_1 \cong H_2$.
(For technical reasons, the algorithm does not exactly implement this strategy, but closely follows the general idea.)

In order to realize this recursive strategy, it remains to ensure that the algorithm makes progress when performing a recursive call.
Actually, this turns out to be a non-trivial task.
Indeed, it may happen that $D_i = X_i$, there is only a single component $Z_1^i$ of $G_i - D_i$, and $N_{G_i}(Z_1^i) = D_i$.
To circumvent this problem, the idea is to compute an isomorphism-invariant extension $\gamma(X_i) \supsetneq X_i$ such that $|\gamma(X_i)| \leq h^4$.
Assuming such an extension can be computed, we simply extend the set $X_i$ until the algorithm arrives in a situation where the recursive scheme discussed above makes progress.
Observe that this is guaranteed to happen as soon as $|X_i| \geq h$ building on Lemma \ref{la:small-separator}.
Also note that we can still artificially individualize all vertices from $X_i$ at a cost of $2^{\CO(h^4\log h)}$ (since any isomorphism can only map vertices from $X_1$ to vertices from $X_2$).

To compute the extension, we exploit the fact that $G_i$ is $(h-1,1)$-WL-bounded by \cite[Corollary 6.3]{Neuen22b} (after individualizing $3$ vertices).
Simply speaking, for every choice of $3$ distinct vertices in $X_i$, after individualizing these vertices and performing the $1$-dimensional Weisfeiler-Leman algorithm, we can identify a color class of size at most $h-1$ to be added to the set $X_i$.
Overall, assuming $|X_i| \leq h$, this gives an extension $\gamma(X_i)$ of size at most $h+h^3(h-1) \leq h^4$.

To implement this high-level strategy, we first prove Lemma \ref{la:small-separator} in the next section.
Afterwards, we compute the entire decompositions of the input graphs in Section \ref{sec:decomposition}.
Finally, the dynamic programming strategy along the computed decomposition is realized in Section \ref{sec:main-algorithm}.

\section{Finding Disjoint and Connected Subgraphs}
\label{sec:disjoint-subtrees}

In this section, we give a proof of Lemma \ref{la:small-separator}.
Towards this end, we first require some additional notation and basic tools for the $2$-dimensional Weisfeiler-Leman algorithm.

Let $G$ be a graph and let $\chi \coloneqq \WL{2}{G}$ be the coloring computed by the $2$-dimensional Weisfeiler-Leman algorithm.
We denote by $C_V = C_V(G,\chi) \coloneqq \{\chi(v,v) \mid v \in V(G)\}$ the set of \emph{vertex colors} under the coloring $\chi$.
Also, for $c \in C_V$, $V_c \coloneqq \{v \in V(G) \mid \chi(v,v) = c\}$ denotes the set of all vertices of color $c$.

Next, consider a set of colors $C \subseteq \{\chi(v,w) \mid v \neq w\}$.
We define the graph $G[C]$ with edge set
\[E(G[C]) \coloneqq \{v_1v_2 \mid \WL{2}{G}(v_1,v_2) \in C\}\]
and vertex set
\[V(G[C]) \coloneqq \bigcup_{e \in E(G[C])}e.\]
In case $C = \{c\}$ consists of a single color we also write $G[c]$ instead of $G[\{c\}]$.

Moreover, let $A_1,\dots,A_\ell$ be the vertex sets of the connected components of $G[C]$.
We also define the graph $G/C$ as the graph obtained from contracting every set $A_i$ to a single vertex.
Formally,
\[V(G/C) \coloneqq \{\{v\} \mid v \in V(G) \setminus V(G[C])\} \cup \{A_1,\dots,A_\ell\}\]
and edge set
\[E(G/C) \coloneqq \{X_1X_2 \mid \exists v_1 \in X_1,v_2 \in X_2\colon v_1v_2 \in E(G)\}.\]
As before, if $C = \{c\}$ consists of a single edge-color we write $G/c$ instead of $G/C$.

\begin{lemma}[see {\cite[Theorem 3.1.11]{ChenP19}}]
 \label{la:factor-graph-2-wl}
 Let $G$ be a graph and $C \subseteq \{\chi(v,w) \mid v \neq w\}$ a set of colors of the stable coloring $\chi \coloneqq \WL{2}{G}$.
 Define
 \[(\chi/C)(X_1,X_2) \coloneqq \{\!\{\chi(v_1,v_2) \mid v_1 \in X_1, v_2 \in X_2\}\!\}\]
 for all $X_1,X_2 \in V(G/C)$.
 Then $\chi/C$ is a stable coloring of the graph $G/C$ with respect to the $2$-dimensional Weisfeiler-Leman algorithm.
 
 Moreover, for all $X_1,X_2,X_1',X_2' \in V(G/C)$, it holds that $(\chi/C)(X_1,X_2) = (\chi/C)(X_1',X_2')$ or $(\chi/C)(X_1,X_2) \cap (\chi/C)(X_1',X_2') = \emptyset$.
\end{lemma}

Next, let $X$ be a set and let $\CP$ be a partition of $X$.
We define the corresponding equivalence relation $\sim_\CP$ on the set $X$ via $x \sim_\CP y$ if there is set $P \in \CP$ such that $x,y \in P$.
Observe that the equivalence classes of $\sim_\CP$ are exactly the classes of $\CP$.
Also, for a second partition $\CQ$ of $X$, we write $\CP \preceq \CQ$ if the partition $\CP$ is finer than $\CQ$, i.e., for every $P \in \CP$ there is some $Q \in \CQ$ such that $P \subseteq Q$.

Assuming $X \subseteq V(G)$, we also define $G/\CP$ to be the graph obtained from $G$ by contracting all blocks $P \in \CP$ into single vertices. 

Let $c \in C_V$.
We say a partition $\CP$ of the set $V_c$ is \emph{$\chi$-definable} if there is a set of colors $C_\CP \subseteq \{\chi(v,w) \mid v \neq w \in V_c\}$ such that
\[v \sim_\CP w \;\;\Leftrightarrow \chi(v,w) \in C_\CP\]
for all $v,w \in V_c$.
Observe that $\chi/C_\CP$ is a $2$-stable coloring for $G/\CP$ in this case by Lemma \ref{la:factor-graph-2-wl}.

\begin{lemma}
 \label{la:partition-overlap}
 Let $G$ be a graph and $\chi$ a $2$-stable coloring.
 Let $c \in C_V$ and let $\CP,\CQ$ be two $\chi$-definable partitions of $V_c$ such that $|\CP| \leq |\CQ|$.
 Let $P_1,\dots,P_\ell \in \CP$ be distinct blocks.
 Then there are $v_i \in P_i$, $i \in [\ell]$, such that $v_i \not\sim_\CQ v_j$ for all distinct $i,j \in [\ell]$.
\end{lemma}

\begin{proof}
 Consider the bipartite graph $B = (\CP,\CQ,E_B)$ where
 \[E_B \coloneqq \{PQ \mid P \in \CP, Q \in \CQ, P \cap Q \neq \emptyset\}.\]

 \begin{claim}
  $B$ is biregular.
 \end{claim}
 \begin{claimproof}
  Let $C_\CP,C_\CQ \subseteq \{\chi(v,w) \mid v \neq w \in V_c\}$ be the color sets defining $\CP$ and $\CQ$, respectively.
  Let $P_1,P_2 \in \CP$ and pick arbitrary elements $v_1 \in P_1$ and $v_2 \in P_2$.
  Since $\chi$ is $2$-stable and $\chi(v_1,v_1) = \chi(v_2,v_2)$, we conclude that
  \[p \coloneqq |P_1| = \{v_1' \in V(G) \mid \chi(v_1,v_1') \in C_\CP\} = \{v_2' \in V(G) \mid \chi(v_2,v_2') \in C_\CP\} = |P_2|.\]
  Similarly, $q \coloneqq |Q_1| = |Q_2|$ for all $Q_1,Q_2 \in \CQ$ and
  \[r \coloneqq |P_1 \cap Q_1| = |P_2 \cap Q_2|\]
  for all $P_1,P_2 \in \CP$ and $Q_1,Q_2 \in \CQ$ such that $P_1 \cap Q_1 \neq \emptyset$ and $P_2 \cap Q_2 \neq \emptyset$.
  It follows that
  \[\deg_B(P_1) = \frac{p}{r} = \deg_B(P_2)\]
  for all $P_1,P_2 \in \CP$, and
  \[\deg_B(Q_1) = \frac{q}{r} = \deg_B(Q_2)\]
  for all $Q_1,Q_2 \in \CQ$.
 \end{claimproof}

 Let $d_\CP \coloneqq \deg_B(P)$ for $P \in \CP$, and $d_\CQ \coloneqq \deg_B(Q)$ for $Q \in \CQ$.
 Then $|E_B| = d_\CP \cdot |\CP| = d_\CQ \cdot |\CQ|$.
 Let $\CP' \subseteq \CP$.
 Then $d_\CQ \cdot |N_B(\CP')|\geq |\CP'| \cdot d_\CP$ and hence, $|N_B(\CP')| \geq \frac{d_\CP}{d_\CQ}|\CP'| = \frac{|\CQ|}{|\CP|}|\CP'| \geq |\CP'|$.
 So by Hall's Marriage Theorem, there is a perfect matching $M$ of $B$.
 Let $Q_1,\dots,Q_\ell \in \CQ$ be those sets matched to $P_1,\dots,P_\ell$ by the matching $M$.
 Pick an arbitrary element $v_i \in P_i \cap Q_i$ for all $i \in [\ell]$.
 Clearly, $v_i \not\sim_\CQ v_j$ for all distinct $i,j \in [\ell]$.
\end{proof}

Let $G$ be a graph and let $\chi$ be a $2$-stable coloring.
We define the graph $G[[\chi]]$ with vertex set $V(G[[\chi]]) \coloneqq C_V(G,\chi)$ and edges
\[E(G[[\chi]]) \coloneqq \{c_1c_2 \mid \exists v_1 \in V_{c_1}, v_2 \in V_{c_2} \colon v_1v_2 \in E(G)\}.\]

Having established all the necessary basic tools, we are now ready to prove the main technical statement of this section.

\begin{lemma}
 \label{la:disjoint-trees}
 Let $G$ be a graph and let $\chi$ be a $2$-stable coloring.
 Suppose that $G[[\chi]]$ is connected and $|V_c| \geq 3$ for every $c \in C_V$.
 Then there are vertex-disjoint, connected subgraphs $H_1,H_2,H_3 \subseteq G$ such that $V(H_r) \cap V_c \neq \emptyset$ for all $r \in \{1,2,3\}$ and $c \in C_V$.
\end{lemma}

\begin{proof}
 Let $F$ be a spanning tree of $G[[\chi]]$.
 Let $c_0 \in C_V$ be an arbitrary color and fix $c_0$ as the root of the tree $F$.
 Also, we direct all edges in $F$ away from the root.
 
 We denote by $L(F)$ the set of leaves of $F$, i.e., all vertices without outgoing edges.
 Also, for $c \neq c_0$ we denote by $\parent(c)$ the unique parent node in the tree $F$.
 Finally, for three graphs $H_1,H_2,H_3 \subseteq G$, we denote by $V(H_1,H_2,H_3) \coloneqq V(H_1) \cup V(H_2) \cup V(H_3)$.
 
 We prove the following statement by induction on $|V(F)| = |C_V|$.
 The input consists of the following objects:
 \begin{enumerate}[label = (\roman*)]
  \item\label{item:disjoint-trees-input-1} a graph $G$ and a $2$-stable coloring $\chi$ such that $|V_c| \geq 3$ for every $c \in C_V$,
  \item\label{item:disjoint-trees-input-2} a (rooted) spanning tree $F$ for $G[[\chi]]$,
  \item\label{item:disjoint-trees-input-3} a collection $(\CP_c)_{c \in C_R}$ of $\chi$-definable partitions $\CP_c = \{P_1^c,P_2^c\}$ (i.e., each partition consists of exactly two blocks) where $C_R \subseteq L(F)$ (i.e., only leaves of $F$ may be equipped with an additional partition), and
  \item\label{item:disjoint-trees-input-4} a collection $(f_c)_{c \in C_R}$ of functions $f_c$ that map each triple $v_1,v_2,v_3 \in V_c$ of distinct vertices such that $1 \leq |\{v_1,v_2,v_3\} \cap P_1^c| \leq 2$ to a partition $f_c(v_1,v_2,v_3) = \CQ_c^{v_1,v_2,v_3} = (Q_1,Q_2,Q_3)  \preceq \CP_c$ into three blocks such that $v_r \in Q_r$ for all $r \in \{1,2,3\}$.
 \end{enumerate}
 Then there are $H_1,H_2,H_3 \subseteq G$ such that
 \begin{enumerate}[label = (\alph*)]
  \item\label{item:disjoint-trees-output-1} $H_1,H_2,H_3$ are pairwise vertex-disjoint,
  \item\label{item:disjoint-trees-output-2} $V(H_r) \cap V_c \neq \emptyset$ for all $c \in C_V$,
  \item\label{item:disjoint-trees-output-3} $|V(H_r) \cap V_c| = 1$ for all leaves $c \in L(F)$,
  \item\label{item:disjoint-trees-output-4} $|V(H_1,H_2,H_3) \cap P| \leq 2$ for all $P \in \CP_c$ and $c \in C_R$, and
  \item\label{item:disjoint-trees-output-5} $H_r'$ is connected where $V(H_r') \coloneqq V(H_r) \cup \bigcup_{c \in C_R}Q_r^c$, and $f_c(v_1^c,v_2^c,v_3^c) = (Q_1^c,Q_2^c,Q_3^c)$, and $\{v_r^c\} = V(H_r) \cap V_c$ for all $c \in C_R$, and
   \[E(H_r') \coloneqq E(G[V(H_r')]) \cup \bigcup_{c \in C_R} \binom{Q_r^c}{2}.\]
 \end{enumerate}
 
 \begin{figure}
  \begin{subfigure}{\textwidth}
  \centering
  \begin{tikzpicture}[scale = 1.2]
   \draw[line width = 1.6pt, color = mRed] (0,0) ellipse (1.8cm and 0.6cm);
   \draw[line width = 1.6pt, color = mBlue] (4,0) ellipse (1.8cm and 0.6cm);

   \draw[line width = 1.6pt, color = mYellow] (2,2) ellipse (1.8cm and 0.6cm);
   \draw[line width = 1.6pt, color = mGreen] (6,2) ellipse (1.8cm and 0.6cm);

   \draw[line width = 1.6pt, color = mTurquoise] (4,4) ellipse (1.8cm and 0.6cm);

   \node at (5.7,4.5) {$c_0$};
   \node at (0.3,2.5) {$c_1$};
   \node at (7.7,2.5) {$c_2$};
   \node at (-1.7,-0.5) {$c_3$};
   \node at (5.7,-0.5) {$c_4$};

   \node[vertex, fill = mRed] (r1) at (-1.2,0) {};
   \node[vertex, fill = mRed] (r2) at (-0.4,0) {};
   \node[vertex, fill = mRed] (r3) at (0.4,0) {};
   \node[vertex, fill = mRed] (r4) at (1.2,0) {};

   \node[vertex, fill = mBlue] (b1) at (3.2,0) {};
   \node[vertex, fill = mBlue] (b2) at (4.0,0) {};
   \node[vertex, fill = mBlue] (b3) at (4.8,0) {};

   \node[vertex, fill = mYellow] (y1) at (0.75,2) {};
   \node[vertex, fill = mYellow] (y2) at (1.25,2) {};
   \node[vertex, fill = mYellow] (y3) at (1.75,2) {};
   \node[vertex, fill = mYellow] (y4) at (2.25,2) {};
   \node[vertex, fill = mYellow] (y5) at (2.75,2) {};
   \node[vertex, fill = mYellow] (y6) at (3.25,2) {};

   \node[vertex, fill = mGreen] (g1) at (4.8,2) {};
   \node[vertex, fill = mGreen] (g2) at (5.6,2) {};
   \node[vertex, fill = mGreen] (g3) at (6.4,2) {};
   \node[vertex, fill = mGreen] (g4) at (7.2,2) {};

   \node[vertex, fill = mTurquoise] (t1) at (2.75,4) {};
   \node[vertex, fill = mTurquoise] (t2) at (3.25,4) {};
   \node[vertex, fill = mTurquoise] (t3) at (3.75,4) {};
   \node[vertex, fill = mTurquoise] (t4) at (4.25,4) {};
   \node[vertex, fill = mTurquoise] (t5) at (4.75,4) {};
   \node[vertex, fill = mTurquoise] (t6) at (5.25,4) {};

   \begin{pgfonlayer}{background}
    \foreach \v/\w in {r1/y1,
                       b1/y1,
                       y1/t1,
                       t1/g1}{
     \draw[line width = 6.0pt, mOrange!80] (\v.center) edge (\w.center);
    }

    \foreach \v/\w in {r2/y2,r2/y3,
                       b2/y2,
                       y2/t2,y3/t3,
                       t2/g3}{
     \draw[line width = 6.0pt, mViolet!80] (\v.center) edge (\w.center);
    }

    \foreach \v/\w in {r3/y4,r3/y5,r3/y6,
                       b3/y6,
                       y4/t4,y5/t5,y6/t6,
                       t5/g4}{
     \draw[line width = 6.0pt, mDarkBlue!80] (\v.center) edge (\w.center);
    }
   \end{pgfonlayer}

   \foreach \v/\w in {r1/y1,r1/y2,r1/y3,r2/y1,r2/y2,r2/y3,r3/y4,r3/y5,r3/y6,r4/y4,r4/y5,r4/y6,
                      b1/y1,b1/y4,b2/y2,b2/y5,b3/y3,b3/y6,
                      y1/t1,y1/t3,y2/t1,y2/t2,y3/t2,y3/t3,y4/t4,y4/t6,y5/t4,y5/t5,y6/t5,y6/t6,
                      t1/g1,t1/g2,t2/g1,t2/g3,t3/g2,t3/g3,t4/g3,t4/g4,t5/g2,t5/g4,t6/g1,t6/g4}{
    \draw[thick] (\v) edge (\w);
   }
  \end{tikzpicture}
  \caption{Since the edge $(c_0,c_2) \in E(F)$ allows expansion, the subgraphs from (b) can be easily extended.}
  \label{fig:find-subgraphs-a}
  \end{subfigure}

  \begin{subfigure}{0.55\textwidth}
  \centering
  \scalebox{0.85}{
  \begin{tikzpicture}
   \draw[line width = 1.6pt, color = mRed] (0,0) ellipse (1.8cm and 0.6cm);
   \draw[line width = 1.6pt, color = mBlue] (4,0) ellipse (1.8cm and 0.6cm);

   \draw[line width = 1.6pt, color = mYellow] (2,2) ellipse (1.8cm and 0.6cm);

   \draw[line width = 1.6pt, color = mTurquoise] (4,4) ellipse (1.8cm and 0.6cm);

   \node at (5.7,4.5) {$c_0$};
   \node at (0.3,2.5) {$c_1$};
   \node at (-1.7,-0.5) {$c_3$};
   \node at (5.7,-0.5) {$c_4$};

   \node[vertex, fill = mRed] (r1) at (-1.2,0) {};
   \node[vertex, fill = mRed] (r2) at (-0.4,0) {};
   \node[vertex, fill = mRed] (r3) at (0.4,0) {};
   \node[vertex, fill = mRed] (r4) at (1.2,0) {};

   \node[vertex, fill = mBlue] (b1) at (3.2,0) {};
   \node[vertex, fill = mBlue] (b2) at (4.0,0) {};
   \node[vertex, fill = mBlue] (b3) at (4.8,0) {};

   \node[vertex, fill = mYellow] (y1) at (0.75,2) {};
   \node[vertex, fill = mYellow] (y2) at (1.25,2) {};
   \node[vertex, fill = mYellow] (y3) at (1.75,2) {};
   \node[vertex, fill = mYellow] (y4) at (2.25,2) {};
   \node[vertex, fill = mYellow] (y5) at (2.75,2) {};
   \node[vertex, fill = mYellow] (y6) at (3.25,2) {};

   \node[vertex, fill = mTurquoise] (t1) at (2.75,4) {};
   \node[vertex, fill = mTurquoise] (t2) at (3.25,4) {};
   \node[vertex, fill = mTurquoise] (t3) at (3.75,4) {};
   \node[vertex, fill = mTurquoise] (t4) at (4.25,4) {};
   \node[vertex, fill = mTurquoise] (t5) at (4.75,4) {};
   \node[vertex, fill = mTurquoise] (t6) at (5.25,4) {};

   \begin{pgfonlayer}{background}
    \foreach \v/\w in {r1/y1,
                       b1/y1,
                       y1/t1}{
     \draw[line width = 6.0pt, mOrange!80] (\v.center) edge (\w.center);
    }

    \foreach \v/\w in {r2/y2,r2/y3,
                       b2/y2,
                       y2/t2,y3/t3}{
     \draw[line width = 6.0pt, mViolet!80] (\v.center) edge (\w.center);
    }

    \foreach \v/\w in {r3/y4,r3/y5,r3/y6,
                       b3/y6,
                       y4/t4,y5/t5,y6/t6}{
     \draw[line width = 6.0pt, mDarkBlue!80] (\v.center) edge (\w.center);
    }
   \end{pgfonlayer}

   \foreach \v/\w in {r1/y1,r1/y2,r1/y3,r2/y1,r2/y2,r2/y3,r3/y4,r3/y5,r3/y6,r4/y4,r4/y5,r4/y6,
                      b1/y1,b1/y4,b2/y2,b2/y5,b3/y3,b3/y6,
                      y1/t1,y1/t3,y2/t1,y2/t2,y3/t2,y3/t3,y4/t4,y4/t6,y5/t4,y5/t5,y6/t5,y6/t6}{
    \draw[thick] (\v) edge (\w);
   }
  \end{tikzpicture}
  }
  \caption{We remove the leaves $c_3$ and $c_4$, and add the color $c_1$ to the set $C_R$.
  For example, we set $f(v_1,v_2,v_4) = (\{v_1\},\{v_2,v_3\},\{v_4,v_5,v_6\})$.
  The vertices from $V_{c_3}$ are then used to connect the parts inductively obtained from (c).}
  \label{fig:find-subgraphs-b}
  \end{subfigure}
  \hfill
  \begin{subfigure}{0.42\textwidth}
  \centering
  \scalebox{0.85}{
  \begin{tikzpicture}
   \draw[line width = 3.2pt, color = mGray] (2,2.6) -- (2,1.4);
   \draw[line width = 1.6pt, color = mYellow] (2,2) ellipse (1.8cm and 0.6cm);

   \draw[line width = 1.6pt, color = mTurquoise] (4,4) ellipse (1.8cm and 0.6cm);

   \node at (5.7,4.5) {$c_0$};
   \node at (0.3,2.5) {$c_1$};

   \node[vertex, fill = mYellow, label = {[label distance = -2pt]below:{\scriptsize $v_1$}}] (y1) at (0.75,2) {};
   \node[vertex, fill = mYellow, label = {[label distance = -2pt]below:{\scriptsize $v_2$}}] (y2) at (1.25,2) {};
   \node[vertex, fill = mYellow, label = {[label distance = -2pt]below:{\scriptsize $v_3$}}] (y3) at (1.75,2) {};
   \node[vertex, fill = mYellow, label = {[label distance = -2pt]below:{\scriptsize $v_4$}}] (y4) at (2.25,2) {};
   \node[vertex, fill = mYellow, label = {[label distance = -2pt]below:{\scriptsize $v_5$}}] (y5) at (2.75,2) {};
   \node[vertex, fill = mYellow, label = {[label distance = -2pt]below:{\scriptsize $v_6$}}] (y6) at (3.25,2) {};

   \node[vertex, fill = mTurquoise] (t1) at (2.75,4) {};
   \node[vertex, fill = mTurquoise] (t2) at (3.25,4) {};
   \node[vertex, fill = mTurquoise] (t3) at (3.75,4) {};
   \node[vertex, fill = mTurquoise] (t4) at (4.25,4) {};
   \node[vertex, fill = mTurquoise] (t5) at (4.75,4) {};
   \node[vertex, fill = mTurquoise] (t6) at (5.25,4) {};

   \begin{pgfonlayer}{background}
    \foreach \v/\w in {y1/t1}{
     \draw[line width = 6.0pt, mOrange!80] (\v.center) edge (\w.center);
    }

    \foreach \v/\w in {y2/t2}{
     \draw[line width = 6.0pt, mViolet!80] (\v.center) edge (\w.center);
    }
    \draw[dashed, line width = 6.0pt, mViolet!80] (y3.center) edge (t3.center);
    \draw[mViolet!80, fill = mViolet!80] (t3.center) circle (8pt);

    \foreach \v/\w in {y4/t4}{
     \draw[line width = 6.0pt, mDarkBlue!80] (\v.center) edge (\w.center);
    }
    \draw[dashed, line width = 6.0pt, mDarkBlue!80] (y5.center) edge (t5.center);
    \draw[dashed, line width = 6.0pt, mDarkBlue!80] (y6.center) edge (t6.center);
    \draw[mDarkBlue!80, fill = mDarkBlue!80] (t5.center) circle (8pt);
    \draw[mDarkBlue!80, fill = mDarkBlue!80] (t6.center) circle (8pt);
   \end{pgfonlayer}

   \foreach \v/\w in {y1/t1,y1/t3,y2/t1,y2/t2,y3/t2,y3/t3,y4/t4,y4/t6,y5/t4,y5/t5,y6/t5,y6/t6}{
    \draw[thick] (\v) edge (\w);
   }
  \end{tikzpicture}
  }
  \caption{Note that $c_1 \in C_R$ and $\CP_{c_1} = \{\{v_1,v_2,v_3\},\{v_4,v_5,v_6\}\}$.
  The subgraphs $H_2$ and $H_3$ (in violet and blue, respectively) contain isolated vertices, which are connected to the rest of the subgraph only in the future using the dashed edges (which are not yet part of $H_2$ and $H_3$).}
  \label{fig:find-subgraphs-c}
  \end{subfigure}
  \caption{Visualization for the construction of $H_1,H_2,H_3$ in Lemma \ref{la:disjoint-trees}.
  The sets $E(H_1),E(H_2),E(H_3)$ are marked in orange, violet and blue, respectively.
  Observe that the edge $(c_0,c_2) \in E(F)$ allows expansion which means the color class $c_2 \in V(F)$ is removed first in the inductive process.
  Afterwards, the leaves $c_3$ and $c_4$ are removed, and the color $c_1$ is added to the set $C_R$ since $G[V_{c_1},V_{c_3}]$ is isomorphic to $2K_{2,3}$.}
  \label{fig:find-subgraphs}
 \end{figure}

 Before diving into the proof, let us give some intuition on this complicated statement.
 First of all, observe that Requirements \ref{item:disjoint-trees-input-1} and \ref{item:disjoint-trees-input-2} cover the prerequisites of the lemma, and Properties \ref{item:disjoint-trees-output-1} and \ref{item:disjoint-trees-output-2} provide two of the three guarantees listed in the lemma.
 Also, for $C_R = \emptyset$, we recover the statement of the lemma (in a slightly stronger form due to Property \ref{item:disjoint-trees-output-3}).

 Now, the basic idea of the inductive proof is to consider a leaf $c$ and construct the graphs $H_1,H_2,H_3$ from a solution $H_1',H_2',H_3'$ obtained by induction after removing all vertices of color $c$ (i.e., removing the leaf $c$ from the tree $F$).
 Let $d$ denote the parent of $c$.
 The proof distinguishes between several cases depending on the on structure of the subgraph $G[V_d,V_c]$, and it relies on roughly three mechanisms for constructing the graphs $H_1,H_2,H_3$ from $H_1',H_2',H_3'$.
 To outline the first two mechanisms, let us assume that $C_R = \emptyset$ for simplicity.

 The most simple case occurs if, for every distinct $v_1,v_2,v_3 \in V_d$, there are distinct $w_1,w_2,w_3 \in V_c$ such that $v_rw_r \in E(G)$ for all $r \in \{1,2,3\}$.
 In this case, we can directly extend the subgraphs $H_1',H_2',H_3'$ by picking arbitrary vertices $v_r \in V_d \cap V(H_r')$, $r \in \{1,2,3\}$, and adding vertex $w_r$ as well as the edge $v_rw_r$ to obtain $H_r$ (see Figure \ref{fig:find-subgraphs-a}).

 Now, suppose that $G[V_d,V_c]$ is isomorphic to a disjoint union of $q$ stars $K_{1,h}$ such that $|V_c| = q$ and $|V_d| = q \cdot h$ (see Figure \ref{fig:c-twins}).
 In this situation, the previous mechanism is not applicable.
 Instead, we contract each connected component of $G[V_d,V_c]$ to a single vertex (i.e., we contract the edge $(d,c)$ in the tree $F$), and obtain $H_1,H_2,H_3$ from $H_1',H_2',H_3'$ by simply uncontracting any contracted vertices contained in $H_1',H_2',H_3'$.
 We remark that, for technical reasons, this mechanism is implemented in a slightly different way in the actual proof.

 Finally, a particularly difficult case occurs if $G[V_d,V_c]$ is isomorphic to $2K_{2,h}$, i.e., it is the disjoint union of two copies of $K_{2,h}$, where $|V_c| = 4$ and $h \geq 3$ (see Figure \ref{fig:find-subgraphs-b}).
 Here, it is not possible to contract the two components since the corresponding color class in the contracted graph would only contain two elements (which violates Requirement \ref{item:disjoint-trees-input-1}).
 Also, it may happen that $V(H_1',H_2',H_3') \cap V_d$ is a subset of one of the two components of $G[V_d,V_c]$ (which prevents us from using the first mechanism).
 To resolve this problem, we introduce a partition $\CP_d$ which propagates the additional requirement that both connected components must be covered by $V(H_1',H_2',H_3')$ up the tree $F$ (see Property \ref{item:disjoint-trees-output-4}), which allows us to employ the first mechanism again.
 Unfortunately, this additional requirement introduces further challenges for other cases.
 To compensate for this, we exploit the fact that we can potentially use vertices from $V_c$ to connect parts from $H_r$ that can not be connected in another way.
 To be more precise, if $V(H_r') \cap V_d$ were to contain vertices $v_r,v_r'$ from the same connected component of $G[V_d,V_c]$, it would be acceptable if these vertices were not connected in $H_r'$ since, as soon as we add some vertex $w_r \in V_c$, this vertex will establish a connection between $v_r$ and $v_r'$.
 This simple idea is captured by Requirement \ref{item:disjoint-trees-input-4} and Property \ref{item:disjoint-trees-output-5}.
 If we construct $H_r'$ in such a way that $V(H_r') \cap V_d = \{v_r\}$, then there is a partition $\CQ_d^{v_1,v_2,v_3} = (Q_1,Q_2,Q_3)  \preceq \CP_d$ which has the property that all vertices within one block may be connected by a vertex from the class $V_c$.
 Hence, it is sufficient if $H_r'$ is connected after adding all the vertices from $Q_r$ to $H_r'$ and turning $Q_r$ into a clique.
 Now, this is precisely what Property \ref{item:disjoint-trees-output-5} guarantees.
 
 \medskip
 Now, let us turn to the inductive proof of the above statement.
 The base case $|C_V| = 1$ is trivial.
 Each of the graphs $H_r$, $r \in \{1,2,3\}$, consists of a single vertex and one can clearly ensure that all properties are satisfied.
 
 For the inductive step we distinguish several cases.
 Let $c$ be a leaf of $F$ and let $d$ be its parent.
 We say that the edge $(d,c)$ \emph{allows expansion} if
 \begin{enumerate}
  \item\label{item:allow-expansion-1} $|N_G(U) \cap V_c| \geq |U|$ for all $U \subseteq V_d$ for which $|U| \leq 3$, and
  \item\label{item:allow-expansion-2} if $c \in C_R$ then $N_G(v) \cap P_1^c \neq \emptyset$ as well as $N_G(v) \cap P_2^c \neq \emptyset$ for all $v \in V_d$.
 \end{enumerate}
 
 \begin{claim}
  \label{claim:expand-trees}
  Suppose $(d,c) \in E(F)$ allows expansion.
  Let $v_1,v_2,v_3 \in V_{d}$ be distinct vertices.
  Then there are distinct $w_1,w_2,w_3 \in V_{c}$ such that $v_rw_r \in E(G)$ for all $r \in \{1,2,3\}$.
  Moreover, if $c \in C_R$ then $|\{w_1,w_2,w_3\} \cap P_1^c| \leq 2$ and $|\{w_1,w_2,w_3\} \cap P_2^c| \leq 2$.
 \end{claim}
 \begin{claimproof}
  Consider the bipartite graph $B \coloneqq G[\{v_1,v_2,v_3\},V_c]$.
  Since $(d,c)$ allows expansion the graph $B$ satisfies the requirements of Hall's Marriage Theorem.
  Hence, there are $w_1,w_2,w_3 \in V_c$ such that $v_rw_r \in E(G)$ for all $r \in \{1,2,3\}$.
  
  So it only remains to ensure the second property.
  Suppose that $c \in C_R$ and $w_1,w_2,w_3$ are in the same block of $\CP_c$ (otherwise there is nothing to be done).
  Without loss of generality suppose that $\{w_1,w_2,w_3\} \subseteq P_1^c$.
  Then one can simply replace $w_1$ by an arbitrary vertex from $N_G(v_1) \cap P_2^c$.
  Observe that this set can not be empty since $(d,c)$ allows expansion.
 \end{claimproof}
 
 First, suppose there is a leaf $c$, and its parent $d$, such that $(d,c)$ allows expansion.
 Let $G' \coloneqq G - V_c$ and $F' \coloneqq F - c$.
 Also, let $C_R' \coloneqq C_R \setminus \{c\}$.
 Clearly,  the input $(G',F',(\CP_c)_{c \in C_R'},(f_c)_{c \in C_R'})$ satisfies the Requirements \ref{item:disjoint-trees-input-1} - \ref{item:disjoint-trees-input-4}.
 By the induction hypothesis, there exists subgraphs $H_1',H_2',H_3' \subseteq G'$ satisfying \ref{item:disjoint-trees-output-1} - \ref{item:disjoint-trees-output-5}.
 We pick arbitrary elements $v_r \in V(H_r') \cap V_d$, $r \in \{1,2,3\}$ (these sets are non-empty by Property \ref{item:disjoint-trees-output-2}).
 Let $w_1,w_2,w_3$ be the vertices provided by Claim \ref{claim:expand-trees}.
 We define the graphs $H_r$, $r \in \{1,2,3\}$, via $V(H_r) \coloneqq V(H_r') \cup \{w_r\}$ and $E(H_r) \coloneqq E(H_r') \cup \{v_rw_r\}$ (see Figure \ref{fig:find-subgraphs-a}).
 It is easy to verify that Properties \ref{item:disjoint-trees-output-1} - \ref{item:disjoint-trees-output-5} are satisfied.
 Observe that, if $c \in C_R$, the graphs are connected even without the additional vertices and edges coming from the partition $f_c(w_1,w_2,w_3)$.
 
 \medskip
 
 Next, consider a color $d \in V(F)$ such that all children of $d$ are leaves in $F$.
 Let $c_1,\dots,c_\ell$ be the children of $d$.
 We assume that $(d,c_i)$ does not allow expansion for all $i \in [\ell]$ (otherwise the previous case is applied).
 
 We start by defining a sequence of $\chi$-definable partitions $\CA_1,\dots,\CA_\ell$ of the set $V_d$.
 Fix some $i \in [\ell]$.
 Since $(d,c_i)$ does not allow expansion, Item \ref{item:allow-expansion-1} or \ref{item:allow-expansion-2} is violated.
 First suppose Item \ref{item:allow-expansion-2} is violated.
 This means that $c_i \in C_R$.
 Let $A_1^i \coloneqq \{v \in V_d \mid N_G(v) \cap V_{c_i} \subseteq P_1^{c_i}\}$ and $A_2^i \coloneqq \{v \in V_d \mid N_G(v) \cap V_{c_i} \subseteq P_2^{c_i}\}$.
 Since Item \ref{item:allow-expansion-2} is violated, we get that $A_1^i \cup A_2^i \neq \emptyset$.
 Since $\chi$ is $2$-stable and $\CP_{c_i}$ is $\chi$-definable, we conclude that $A_1^i \cup A_2^i = V_d$, i.e., $\CA_i \coloneqq \{A_1^i,A_2^i\}$ is a partition of $V_d$.
 Moreover, $|A_1^i| = |A_2^i|$ and $\CA_i$ is $\chi$-definable.
 
 So suppose Item \ref{item:allow-expansion-2} is not violated, meaning Item \ref{item:allow-expansion-1} is violated.
 So there is some set $U_i \subseteq V_d$ such that $|U_i| \leq 3$ and $|N_G(U_i) \cap V_{c_i}| < |U_i|$.
 If $|N_G(v) \cap V_{c_i}| = 1$ for all $v \in V_d$ then $G[V_d,V_c]$ is isomorphic to the disjoint union of $\ell$ copies of $K_{1,h}$ for some $h \geq 2$ and $\ell \geq 3$.
 Otherwise $|N_G(v) \cap V_{c_i}| \geq 2$ for all $v \in V_d$.
 This implies that $|U_i| = 3$ and $|N_G(v) \cap V_{c_i}| = 2$ for all $v \in V_d$.
 Moreover, $N_G(v) \cap V_{c_i} = N_G(v') \cap V_{c_i}$ for all $v,v' \in U_i$.
 We say that $v,v' \in V_d$ are \emph{$c_i$-twins} if $N_G(v) \cap V_{c_i} = N_G(v') \cap V_{c_i}$.
 Let $A_1^i,\dots,A_{q_i}^i$ denote the equivalence classes of $c_i$-twins, $i \in [\ell]$.
 In both cases, $q_i \geq 2$ since $|V_c| \geq 3$.
 Since $\chi$ is $2$-stable we have that $|A_j^i| = |A_{j'}^i|$ for all $j,j' \in [q_i]$.
 Moreover, the partition $\CA_i \coloneqq \{A_1^i,\dots,A_{q_i}^i\}$ is $\chi$-definable.
 This completes the description of the partitions $\CA_1,\dots,\CA_\ell$.
 
 Let $q_i \coloneqq |\CA_i|$ denote the number of blocks of $\CA_i$.
 Without loss of generality assume that $q_1 \leq q_i$ for all $i \in [\ell]$.
 For ease of notation, define $c \coloneqq c_1$, $\CA \coloneqq \CA_1$ and $q \coloneqq q_1$.
 Recall that $q \geq 2$.
 
 \begin{figure}
  \centering
  \begin{tikzpicture}

   \draw[line width = 1.6pt, color = mTurquoise] (0,3) ellipse (2.75cm and 0.6cm);

   \draw[line width = 1.6pt, color = mRed] (-3,0) ellipse (2.0cm and 0.6cm);
   \draw[line width = 1.6pt, color = mBlue] (3,0) ellipse (2.0cm and 0.6cm);

   \node at (-0.6,0) {$c_1$};
   \node at (5.4,0) {$c_2$};
   \node at (3.15,3) {$d$};

   \node[vertex, fill = mTurquoise] (t1) at (-2.25,3) {};
   \node[vertex, fill = mTurquoise] (t2) at (-1.75,3) {};
   \node[vertex, fill = mTurquoise] (t3) at (-1.25,3) {};
   \node[vertex, fill = mTurquoise] (t4) at (-0.5,3) {};
   \node[vertex, fill = mTurquoise] (t5) at (0.0,3) {};
   \node[vertex, fill = mTurquoise] (t6) at (0.5,3) {};
   \node[vertex, fill = mTurquoise] (t7) at (1.25,3) {};
   \node[vertex, fill = mTurquoise] (t8) at (1.75,3) {};
   \node[vertex, fill = mTurquoise] (t9) at (2.25,3) {};

   \node[vertex, fill = mRed] (r1) at (-4,0) {};
   \node[vertex, fill = mRed] (r2) at (-3,0) {};
   \node[vertex, fill = mRed] (r3) at (-2,0) {};

   \node[vertex, fill = mBlue] (b1) at (2,0) {};
   \node[vertex, fill = mBlue] (b2) at (3,0) {};
   \node[vertex, fill = mBlue] (b3) at (4,0) {};

   \foreach \v/\w in {t1/r1,t2/r1,t3/r1,t4/r2,t5/r2,t6/r2,t7/r3,t8/r3,t9/r3,
                      t1/b1,t2/b2,t3/b3,t4/b1,t5/b2,t6/b3,t7/b1,t8/b2,t9/b3}{
    \draw[thick] (\v) edge (\w);
   }

   \begin{pgfonlayer}{background}
    \draw[mGray!60, fill = mGray!60] (-1.75,3) ellipse (0.8cm and 0.3cm);
    \draw[mGray!60, fill = mGray!60] (0,3) ellipse (0.8cm and 0.3cm);
    \draw[mGray!60, fill = mGray!60] (1.75,3) ellipse (0.8cm and 0.3cm);

    \foreach \v/\w in {t1/r1,t2/r1,t3/r1,t1/b1}{
     \draw[line width = 6.0pt, mOrange!80] (\v.center) edge (\w.center);
    }

    \foreach \v/\w in {t4/r2,t5/r2,t6/r2,t5/b2}{
     \draw[line width = 6.0pt, mViolet!80] (\v.center) edge (\w.center);
    }

    \foreach \v/\w in {t7/r3,t8/r3,t9/r3,t9/b3}{
     \draw[line width = 6.0pt, mDarkBlue!80] (\v.center) edge (\w.center);
    }
   \end{pgfonlayer}
  \end{tikzpicture}
  \caption{The figure depicts the case $q \geq 3$ in the proof of Lemma \ref{la:disjoint-trees}.
   The partition $\CA$ is shown in gray and parts of the sets $E(H_1),E(H_2),E(H_3)$ are marked in orange, violet and blue, respectively.}
  \label{fig:c-twins}
 \end{figure}

 We distinguish two cases.
 First suppose that $q \geq 3$ (see Figure \ref{fig:c-twins}).
 Let $F' \coloneqq F - \{c_i \mid i \in [\ell]\}$.
 Observe that $d$ is a leaf of $F'$.
 Also define
 \[G' \coloneqq \left(G - \bigcup_{i \in [\ell]} V_{c_i}\right)/\CA,\]
 i.e., $G'$ is the graph obtained from $G$ by deleting all vertices from $V_{c_i}$, $i \in [\ell]$, and contracting the sets $A_1^1,\dots,A_q^1$ to single vertices.
 Let $C_R' \coloneqq C_R \setminus \{c_i \mid i \in [\ell]\}$.
 Using Lemma \ref{la:factor-graph-2-wl} and $q \geq 3$, the input $(G',F',(\CP_c)_{c \in C_R'},(f_c)_{c \in C_R'})$ satisfies the Requirements \ref{item:disjoint-trees-input-1} - \ref{item:disjoint-trees-input-4}.
 By the induction hypothesis, there are subgraphs $H_1',H_2',H_3' \subseteq G'$ satisfying \ref{item:disjoint-trees-output-1} - \ref{item:disjoint-trees-output-5}.
 Let $\{A_r\} = V(H_r') \cap V_d$, $r \in \{1,2,3\}$ (recall that $|V(H_r') \cap V_d| = 1$ by Property \ref{item:disjoint-trees-output-3}).
 
 \begin{claim}
  \label{claim:expand-modulo-twins}
  Let $i \in [\ell]$ and $v_1,v_2,v_3 \in V_d$ such that $v_r$ and $v_{r'}$ are not $c_i$-twins for all distinct $r,r' \in \{1,2,3\}$.
  Then there are distinct vertices $w_1,w_2,w_3 \in V_{c_i}$ such that $v_rw_r \in E(G)$ for all $r \in \{1,2,3\}$.
  Moreover, if $c_i \in C_R$, then $|\{w_1,w_2,w_3\} \cap P_1^{c_i}| \leq 2$ and $|\{w_1,w_2,w_3\} \cap P_2^{c_i}| \leq 2$.
 \end{claim}
 \begin{claimproof}
  Consider the bipartite graph $B \coloneqq G[\{v_1,v_2,v_3\},V_{c_i}]$.
  Since $v_1,v_2,v_3$ are pairwise no $c_i$-twins, the graph $B$ satisfies the requirements of Hall's Marriage Theorem.
  Hence, there are $w_1,w_2,w_3 \in V_{c_i}$ such that $v_rw_r \in E(G)$ for all $r \in \{1,2,3\}$.
  
  So it only remains to ensure the second property.
  Suppose that $c_i \in C_R$ and $w_1,w_2,w_3$ are in the same block of $\CP_{c_i}$ (otherwise there is nothing to be done).
  Without loss of generality suppose that $\{w_1,w_2,w_3\} \subseteq P_1^{c_i}$.
  Then one can simply replace $w_1$ by an arbitrary vertex from $N_G(v_1) \cap P_2^{c_i}$.
  Observe that this set can not be empty since $(d,c_i)$ satisfies Item \ref{item:allow-expansion-2}.
  Indeed, if it would not satisfy Item \ref{item:allow-expansion-2}, then $q_i = 2$ by construction.
  But $q_i \geq q \geq 3$.
 \end{claimproof}
 
 Now, we construct $H_r$ from $H_r'$ as follows.
 First, we uncontract the set $A_r$, i.e., we remove vertex $A_r$ from $H_r'$ and add all vertices in $A_r$.
 Let $v_r \in A_r$ be an arbitrary vertex, $r \in \{1,2,3\}$.
 By definition, $v_1, v_2$ and $v_3$ are pairwise not $c$-twins.
 Hence, by Claim \ref{claim:expand-modulo-twins}, there are distinct $w_1,w_2,w_3 \in V_{c}$ such that $v_rw_r \in E(G)$ for all $r \in \{1,2,3\}$.
 Since all vertices in $A_r$ are $c$-twins by definition, it holds that $A_r \subseteq N_G(w_r)$ for $r \in \{1,2,3\}$.
 We add $w_r$ as well as all edges $w_rv_r'$, $v_r' \in A_r$, to the graph $H_r$.
 
 Now, it only remains to cover the classes $V_{c_i}$, $2 \leq i \leq \ell$.
 Fix some $2 \leq i \leq \ell$.
 By Lemma \ref{la:partition-overlap}, there are $v_1^i \in A_1$, $v_2^i \in A_2$ and $v_3^i \in A_3$ that are pairwise not $c_i$-twins.
 Thus, by Claim \ref{claim:expand-modulo-twins}, there are distinct $w_1^i, w_2^i,w_3^i \in V_{c_i}$ such that $v_r^iw_r^i \in E(G)$ for all $r \in \{1,2,3\}$.
 We add these vertices and edges to the graph $H_r$.
 This completes the description of $H_r$.
 
 It is easy to see that Properties \ref{item:disjoint-trees-output-1} - \ref{item:disjoint-trees-output-3} are satisfied.
 Property \ref{item:disjoint-trees-output-4} is satisfied if $c_i \in C_R$, $i \in [\ell]$, by Claim \ref{claim:expand-modulo-twins}.
 For colors $c \in C_R \setminus \{c_i \mid i \in [\ell]\}$, Property \ref{item:disjoint-trees-output-4} follows directly from the induction hypothesis.
 Finally, for Property \ref{item:disjoint-trees-output-5}, observe that $H_r[V_d \cup \bigcup_{i \in [\ell]}V_{c_i}]$ is connected for every $r \in \{1,2,3\}$.
 This follows from the fact that $A_r \subseteq N_G(w_r)$.
 
 \medskip

 \begin{figure}
  \centering
  \begin{tikzpicture}
   \draw[line width = 3.2pt, color = mGray] (0,-0.6) -- (0,0.6);
   \draw[line width = 3.2pt, color = mGray] (0,2.4) -- (0,3.6);

   \draw[line width = 1.6pt, color = mTurquoise] (0,3) ellipse (3.2cm and 0.6cm);
   \draw[line width = 1.6pt, color = mRed] (0,0) ellipse (3.2cm and 0.6cm);

   \node at (3.6,0) {$c$};
   \node at (3.6,3) {$d$};

   \node at (-1.6,-0.8) {$P_1^c$};
   \node at (1.6,-0.8) {$P_2^c$};

   \node at (-1.6,3.8) {$A_1^1$};
   \node at (1.6,3.8) {$A_2^1$};

   \node[vertex, fill = mTurquoise] (t1) at (-2.8,3) {};
   \node[vertex, fill = mTurquoise] (t2) at (-2.0,3) {};
   \node[vertex, fill = mTurquoise] (t3) at (-1.2,3) {};
   \node[vertex, fill = mTurquoise] (t4) at (-0.4,3) {};
   \node[vertex, fill = mTurquoise] (t5) at (0.4,3) {};
   \node[vertex, fill = mTurquoise] (t6) at (1.2,3) {};
   \node[vertex, fill = mTurquoise] (t7) at (2.0,3) {};
   \node[vertex, fill = mTurquoise] (t8) at (2.8,3) {};

   \node[vertex, fill = mRed] (r1) at (-2.4,0) {};
   \node[vertex, fill = mRed] (r2) at (-0.8,0) {};
   \node[vertex, fill = mRed] (r3) at (0.8,0) {};
   \node[vertex, fill = mRed] (r4) at (2.4,0) {};

   \foreach \v/\w in {t1/r1,t2/r1,t3/r2,t4/r2,t5/r3,t6/r3,t7/r4,t8/r4}{
    \draw[thick] (\v) edge (\w);
   }

   \begin{pgfonlayer}{background}
    \draw[mGray!60, fill = mGray!60] (-2.4,3) ellipse (0.6cm and 0.3cm);
    \draw[mGray!60, fill = mGray!60] (-0.8,3) ellipse (0.6cm and 0.3cm);
    \draw[mGray!60, fill = mGray!60] (0.8,3) ellipse (0.6cm and 0.3cm);
    \draw[mGray!60, fill = mGray!60] (2.4,3) ellipse (0.6cm and 0.3cm);
   \end{pgfonlayer}
  \end{tikzpicture}
  \caption{The figure depicts the case $q = 2$ and $|N_G(v) \cap V_c| = 1$ for all $v \in V_d$ in the proof of Lemma \ref{la:disjoint-trees}.
   The partition $\CA^*$ consists of four blocks shown in gray.
   We need to contract the four equivalence classes, since otherwise it may happen that $v_1$ and $v_2$ appear in the same block where $\{v_r\} = H_r' \cap V_d$ and $H_r'$ is the subgraphs obtained from the induction hypothesis, $r \in \{1,2,3\}$.
   In this case, only one of the two subgraphs $H_1'$ and $H_2'$ could be extended to color $c$.}
  \label{fig:partition-a-star}
 \end{figure}
 
 So it only remains to cover the case $q = 2$.
 At this point, we need to define another partition $\CA^*$ of $V_d$ to cover a certain special case.
 If $|N_G(v) \cap V_c| = 1$ for all $v \in V_d$ then we define $\CA^*$ to be the partition into the equivalence classes of the $c$-twin relation (see Figure \ref{fig:partition-a-star}).
 Observe that $|\CA^*| = |V_c|$ in this case.
 Otherwise, $\CA^* \coloneqq \{\{v\} \mid v \in V_d\}$ is defined to be the trivial partition.
 In this case $|\CA^*| = |V_d|$.
 Observe that, in both cases, $|\CA^*| \geq 3$, $\CA^*$ is $\chi$-definable and $\CA^* \preceq \CA$.
 In particular, since $\CA^* \preceq \CA$, we can interpret $\CA$ as a partition of $\CA^*$.
 
 Let $F' \coloneqq F - \{c_i \mid i \in [\ell]\}$.
 Observe that $d$ is a leaf of $F'$.
 Also define
 \[G' \coloneqq \left(G - \bigcup_{i \in [\ell]} V_{c_i}\right)/\CA^*,\]
 i.e., $G'$ is the graph obtained from $G$ by deleting all vertices from $V_{c_i}$, $i \in [\ell]$, and contracting all sets $A^* \in \CA^*$ to a single vertex.
 Here, $\CA^*$ forms the class of vertices of color $d$ in the graph $G'$.
 Moreover, let $C_R' \coloneqq (C_R \setminus \{c_i \mid i \in [\ell]\}) \cup \{d\}$.
 Also, let $\CP_d \coloneqq \CA$ (where $\CA$ is interpreted as a partition of $\CA^*$).
 Recall that $|\CA| = q = 2$ and $\CA$ is $\chi$-definable.
 
 It remains to define $f_d$.
 Observe that $f_d$ needs to be defined on triples of vertices $A_1^*,A_2^*,A_3^* \in \CA^*$.
 Let $A_1^*,A_2^*,A_3^* \in \CA^*$ be distinct such that $1 \leq |\{A_1^*,A_2^*,A_3^*\} \cap P_1^d| \leq 2$.
 To define $f_d(A_1^*,A_2^*,A_3^*)$ we distinguish two further subcases.
 
 First, suppose that $(d,c)$ satisfies Item \ref{item:allow-expansion-2}, meaning that it violates Item \ref{item:allow-expansion-1} (recall that $(d,c)$ does not allow expansion).
 Then $G[V_d,V_c]$ is isomorphic to a disjoint union of two copies of $K_{2,h}$ for some number $h \geq 3$, and $|V_c| = 4$ (see Figure \ref{fig:find-subgraphs}).
 In this case, we pick an arbitrary partition according to the requirements of Item \ref{item:disjoint-trees-input-4}.
 
 In the other case, $(d,c)$ violates Item \ref{item:allow-expansion-2}.
 In particular, $c \in C_R$.
 This means that there are only edges between $P_1^d$ and $P_1^c$ as well as between $P_2^d$ and $P_2^c$ (here, we interpret $\CP_d = \CA$ as a partition of $V_d$ in the graph $G$).
 By the definition of $\CA^*$, there are $w_1,w_2,w_3 \in V_c$ such that $E_G(A_r^*,\{w_r\}) \neq \emptyset$, and $|\{w_1,w_2,w_3\} \cap P_1^c| = |\{A_1^*,A_2^*,A_3^*\} \cap P_1^d|$ as well as $|\{w_1,w_2,w_3\} \cap P_2^c| = |\{A_1^*,A_2^*,A_3^*\} \cap P_2^d|$.
 Let $\CQ_c^{w_1,w_2,w_3} = f_c(w_1,w_2,w_3) = (Q_1^c,Q_2^c,Q_3^c)$.
 We define $\CQ_d^{A_1^*,A_2^*,A_3^*} = f_d(A_1^*,A_2^*,A_3^*) = (Q_1^d,Q_2^d,Q_3^d)$ in such a way that $A_r^* \in Q_r^d$, and every $A^* \in Q_r^d$ contains an element that has a neighbor in the set $Q_r^c$ in the graph $G$.
 Clearly, such a partition exists since $G[V_d,V_c]$ is biregular.
 Moreover, it is easy to see that $\CQ_d^{A_1^*,A_2^*,A_3^*} \preceq \CP_d$.
 
 We apply the induction hypothesis to $(G',F',(\CP_c)_{c \in C_R'},(f_c)_{c \in C_R'})$ using Lemma \ref{la:factor-graph-2-wl}.
 This results in subgraphs $H_1',H_2',H_3' \subseteq G'$ satisfying \ref{item:disjoint-trees-output-1} - \ref{item:disjoint-trees-output-5}.
 Let $A_r^*$ be the unique vertex in the set $V(H_r') \cap \CA^*$, $r \in \{1,2,3\}$ (recall that $|V(H_r') \cap \CA^*| = 1$ by Property \ref{item:disjoint-trees-output-3}).
 Observe that $1 \leq |\{A_1^*,A_2^*,A_3^*\} \cap P_1^d| \leq 2$ and $1 \leq |\{A_1^*,A_2^*,A_3^*\} \cap P_2^d| \leq 2$ by Property \ref{item:disjoint-trees-output-4}.
 Also let $f_d(A_1^*,A_2^*,A_3^*) = (Q_1^d,Q_2^d,Q_3^d)$.
 Note that $A_r^* \in Q_r^d$.
 
 We perform several steps to obtain $H_r$ from the graph $H_r'$.
 As a first step, all vertices from the set $\bigcup_{A^* \in Q_r^d}A^*$ are added to the vertex set of $H_r$.
 
 Next, consider the color $c = c_1$.
 We again need to distinguish between the two cases already used for the definition of $f_d$.
 First, suppose that $(d,c)$ satisfies Item \ref{item:allow-expansion-2}, meaning that it violates Item \ref{item:allow-expansion-1}.
 Recall that, in this case, $G[V_d,V_c]$ is isomorphic to a disjoint union of two copies of $K_{2,h}$ for some number $h \geq 3$, and $|V_c| = 4$.
 Also, $\CA^* = \{\{v\} \mid v \in V_d\}$, i.e., there is a natural one-to-one correspondence between $V_d$ and $\CA^*$.
 Let $v_r \in V_d$ such that $A_r^* = \{v_d\}$.
 Observe that $P_1^d$ and $P_2^d$ correspond to the two connected components of $G[V_d,V_c]$ (here, we again interpret $\CP_d = \CA$ as a partition of $V_d$).
 Hence, $v_1,v_2,v_3$ cover both connected components.
 It is easy to see that there are distinct $w_1,w_2,w_3 \in V_c$ such that $v_rw_r \in E(G)$ for all $r \in \{1,2,3\}$.
 We add vertex $w_r$ as well as the edge $v_rw_r$ to the graph $H_r$.
 Recall that $P_1^d$ and $P_2^d$ are defined as the equivalence classes of the $c$-twins relation, and that $\{Q_1^d,Q_2^d,Q_3^d\} \preceq \{P_1^d,P_2^d\}$.
 This implies that $\bigcup_{A^* \in Q_r^d}A^* \subseteq N_G(w_r)$, $r \in \{1,2,3\}$.
 In particular, $H_r[V_d \cup V_c]$ is connected.
 
 In the other case $(d,c)$ violates Item \ref{item:allow-expansion-2}.
 Let $w_1,w_2,w_3$ be the vertices used for the definition of $f_d(A_1^*,A_2^*,A_3^*) = (Q_1^d,Q_2^d,Q_3^d)$.
 Also suppose that $f_c(w_1,w_2,w_3) = (Q_1^c,Q_2^c,Q_3^c)$.
 Recall that $E_G(A_r^*,\{w_r\}) \neq \emptyset$, and $|\{w_1,w_2,w_3\} \cap P_1^c| = |\{A_1^*,A_2^*,A_3^*\} \cap P_1^d|$ as well as $|\{w_1,w_2,w_3\} \cap P_2^c| = |\{A_1^*,A_2^*,A_3^*\} \cap P_2^d|$.
 Again, we add vertex $w_r$ as well as all edges $v_rw_r$, $v_r \in A_r^*$, to the graph $H_r$ (observe that $w_rv_r \in E(G)$ for all $v_r \in A_r^*$ by the definition of the partition $\CA^*$).
 Observe that Property \ref{item:disjoint-trees-output-4} for color $d$ (ensured by the induction hypothesis), together with conditions above, implies Property \ref{item:disjoint-trees-output-4} for color $c$.
 Also observe that $H_r[V_d \cup V_c]$ is connected after adding all vertices from $Q_r^c$ and turning $Q_r^c$ into clique, as all vertices from $\bigcup_{A^* \in Q_r^d}A^*$ have a neighbor in the set $Q_r^c$ by definition.
 
 Now, we turn to the other leaves $c_2,\dots,c_\ell$.
 Observe that, up to this point and ignoring the leaves $c_2,\dots,c_\ell$, the graphs $H_1,H_2,H_3$ satisfy Properties \ref{item:disjoint-trees-output-1} - \ref{item:disjoint-trees-output-5}.
 So let $i \in \{2,\dots,\ell\}$ and consider the leaf $c_i$.
 
 \begin{claim}
  \label{claim:expand-with-two-twins}
  There are distinct vertices $w_1,w_2,w_3 \in V_{c_i}$ such that, for every $r \in \{1,2,3\}$, there is some $v_r' \in \bigcup_{A^* \in Q_r^d}A^*$ such that $v_r'w_r \in E(G)$.
  Moreover, if $c_i \in C_R$, then $|\{w_1,w_2,w_3\} \cap P_1^{c_i}| \leq 2$ and $|\{w_1,w_2,w_3\} \cap P_2^{c_i}| \leq 2$.
 \end{claim}
 \begin{claimproof}
  Let $\widehat{Q}_r \coloneqq \bigcup_{A^* \in Q_r^d}A^*$.
  Observe that $(\widehat{Q}_1,\widehat{Q}_2,\widehat{Q}_3)$ forms a partition of $V_d$ that refines $\CA$.
  Consider the bipartite graph $B = (\{\widehat{Q}_1,\widehat{Q}_2,\widehat{Q}_3\},V_{c_i},E_B)$ where
  \[E_B \coloneqq \{\widehat{Q}_rw \mid \exists v \in \widehat{Q}_r \colon vw \in E(G)\}.\]
  Suppose towards a contradiction that $B$ that does not satisfy Hall's condition, i.e., there is a set $U \subseteq \{\widehat{Q}_1,\widehat{Q}_2,\widehat{Q}_3\}$ such that $|N_B(U)| < |U|$.
  First observe that $N_B(\{\widehat{Q}_1,\widehat{Q}_2,\widehat{Q}_3\}) = V_{c_i}$ and $|V_{c_i}| \geq 3$ by Condition \ref{item:disjoint-trees-input-1}.
  So $U \neq \{\widehat{Q}_1,\widehat{Q}_2,\widehat{Q}_3\}$ which means that $|U| \leq 2$.
  Next, observe that there are no isolated vertices in $B$, i.e., $|N_B(U)| \geq 1$.
  Together, this means that $|U| = 2$ and $|N_B(U)| = 1$ (since $|N_B(U)| < |U|$ by assumption).
  Without loss of generality suppose that $U = \{\widehat{Q}_1,\widehat{Q}_2\}$ and pick $w \in V_{c_i}$ such that $N_B(U) = \{w\}$.
  Then $N_G(v) \cap V_{c_i} \subseteq \{w\}$ for all $v \in \widehat{Q}_1 \cup \widehat{Q}_2$.
  Since $G[V_d,V_{c_i}]$ is biregular and contains at least one edge, it follows that $|N_G(v) \cap V_{c_i}| = 1$ for every $v \in V_d$.
  Hence, $G[V_d,V_{c_i}]$ is isomorphic to $\ell$ disjoint copies of $K_{1,h}$ for some $\ell,h \geq 1$.
  Also, $|V_{c_i}| = \ell$ and $|V_d| = \ell \cdot h$.
  Now, observe that $N_G(v) \cap V_{c_i} = \{w\}$ for all $v \in \widehat{Q}_1 \cup \widehat{Q}_2$.
  This means that $h \geq |\widehat{Q}_1 \cup \widehat{Q}_2| = |\widehat{Q}_1| + |\widehat{Q}_2|$.
  Since $\{\widehat{Q}_1,\widehat{Q}_2,\widehat{Q}_3\} \preceq \CA$ and $\CA$ forms an equipartition into two blocks, we get that $|\widehat{Q}_1| + |\widehat{Q}_2| \geq \frac{1}{2}|V_d|$.
  Together, this implies that $h \geq \frac{1}{2}|V_d|$ and thus, $\ell \leq 2$.
  But this is a contradiction, since $\ell = |V_{c_i}| \geq 3$ by Condition \ref{item:disjoint-trees-input-1}.
  
  So $B$ satisfies Hall's condition which, by Hall's Marriage Theorem, means that $B$ contains a perfect matching $\{\widehat{Q}_1w_1,\widehat{Q}_2w_2,\widehat{Q}_3w_3\}$.
  
  Next, suppose $c_i \in C_R$ and $w_1,w_2,w_3$ are in the same block of $\CP_{c_i}$ (otherwise there is nothing to be done).
  Without loss of generality suppose that $\{w_1,w_2,w_3\} \subseteq P_1^{c_i}$.
  Since $N_B(\{\widehat{Q}_1,\widehat{Q}_2,\widehat{Q}_3\}) = V_{c_i}$ there is some $r \in \{1,2,3\}$ such that $N_B(\widehat{Q}_r) \cap P_2^{c_i} \neq \emptyset$.
  Hence, one can simply replace $w_r$ by an arbitrary element $w_r' \in N_B(\widehat{Q}_r) \cap P_2^{c_i}$.
  This ensures that $|\{w_1,w_2,w_3\} \cap P_1^{c_i}| \leq 2$ and $|\{w_1,w_2,w_3\} \cap P_2^{c_i}| \leq 2$.
  
  To complete the proof, for every $r \in \{1,2,3\}$, we pick some element $v_r' \in \widehat{Q}_r$ such that $v_r'w_r \in E(G)$ (the existence of such elements follows directly from the definition of $E_B$).
 \end{claimproof}
 
 Let $w_1,w_2,w_3$ and $v_1',v_2',v_3'$ be the vertices from Claim \ref{claim:expand-with-two-twins}.
 We add vertex $w_r$ as well as the edge $v_r'w_r$ to the graph $H_r$.
 
 Applying this procedure for all $i \in \{2,\dots,\ell\}$ completes the description of the graphs $H_1,H_2,H_3$.
 Building on the comments above, it is easy to see that the graphs satisfy Properties \ref{item:disjoint-trees-output-1} - \ref{item:disjoint-trees-output-5}.
 This completes the final case.
\end{proof}

Building on the last lemma, we can now prove Lemma \ref{la:small-separator}.

\begin{proof}[Proof of Lemma \ref{la:small-separator}]
 Let $\chi$ be a $2$-stable coloring such that $|[v]_\chi| = 1$ for all $v \in D$ and $|[w]_\chi| \geq 3$ for all $w \in V(G) \setminus D$.
 Suppose towards a contradiction that $|N_G(Z)| \geq h$, and pick $v_1,\dots,v_h \in N_G(Z)$ to be distinct vertices.
 Let $C \coloneqq \{\chi(v,v) \mid v \in Z\}$ be the set of vertex colors appearing in the set $Z$.
 Note that $(G[[\chi]])[C]$ is connected, and $|V_c| \geq 3$ for all $c \in C$.
 Let $W \coloneqq \{w \in V(G) \mid \chi(w,w) \in C\}$.
 Observe that $W \cap D = \emptyset$.
 By Lemma \ref{la:disjoint-trees}, there are connected, vertex-disjoint subgraphs $H_1,H_2,H_3 \subseteq G[W]$ such that $V(H_r) \cap V_c \neq \emptyset$ for all $r \in \{1,2,3\}$ and $c \in C$.
 
 Now let $i \in [h]$.
 Since $v_i \in N_G(Z)$ there is some vertex $w_i \in Z \subseteq W$ such that $v_iw_i \in E(G)$.
 Let $c_i \coloneqq \chi(w_i,w_i)$.
 Observe that $c_i \in C$.
 Also, $V_{c_i} \subseteq N_G(v_i)$ since $|[v_i]_\chi| = 1$ and $\chi$ is $2$-stable.
 This implies that $N_G(v_i) \cap V(H_r) \neq \emptyset$ for all $r \in \{1,2,3\}$, because $V(H_r) \cap V_{c_i} \neq \emptyset$.
 But this results in a minor isomorphic to $K_{3,h}$ with vertices $v_1,\dots,v_h$ on the right side, and vertices $V(H_1),V(H_2),V(H_3)$ on the left side.
\end{proof}

Besides Lemma \ref{la:small-separator}, we also require a second tool which is used to define the extension sets $\gamma(X_i)$ which we needed to ensure the recursive algorithm makes progress.

\begin{lemma}
 \label{la:find-small-color-class}
 Let $G$ be a graph that excludes $K_{3,h}$ as a minor.
 Also let $X \subseteq V(G)$ and define $D \coloneqq \cl_{h-1,1}^G(X)$.
 Let $Z$ be a connected component of $G - D$.
 Then $|N_G(Z)| < 3$.
\end{lemma}

The lemma essentially follows from \cite[Lemma 6.2]{Neuen22b}.
For the sake of completeness and due to its simplicity, a complete proof is still given below.

\begin{proof}
 Let $\chi$ be a $1$-stable coloring such that $|[v]_\chi| = 1$ for all $v \in D$ and $|[w]_\chi| \geq h$ for all $w \in V(G) \setminus D$.
 Suppose towards a contradiction that $|N_G(Z)| \geq 3$, and pick $v_1,v_2,v_3 \in N_G(Z)$ to be distinct vertices.
 Let $C \coloneqq \{\chi(v) \mid v \in Z\}$, and define $H$ to be the graph with $V(H) \coloneqq C$ and
 \[E(H) \coloneqq \{c_1c_2 \mid \exists v_1 \in \chi^{-1}(c_1), v_2 \in \chi^{-1}(v_2) \colon v_1v_2 \in E(G)\}.\]
 Let $T$ be a spanning tree of $H$.
 Also, for each $i \in \{1,2,3\}$, fix a color $c_i \in C$  such that $N_G(v_i) \cap \chi^{-1}(c_i) \neq \emptyset$.
 Let $T'$ be the induced subtree obtained from $T$ by repeatedly removing all leaves distinct from $c_1,c_2,c_3$.
 Finally, let $T''$ be the tree obtained from $T'$ by adding three fresh vertices $v_1,v_2,v_3$ where $v_i$ is connected to $c_i$.
 Observe that $v_1,v_2,v_3$ are precisely the leaves of $T''$.
 Now, $T''$ contains a unique node $c$ of degree three (possibly $c = c_i$ for some $i \in \{1,2,3\}$).
 Observe that $|\chi^{-1}(c)| \geq h$.
 We define $C_i$ to be the set of all internal vertices which appear on the unique path from $v_i$ to $c$ in the tree $T''$.
 Finally, define $U_i \coloneqq \{v_i\} \cup \bigcup_{c' \in C_i} \chi^{-1}(c')$.
 
 Since $\chi$ is $1$-stable and $|[v_i]_{\chi}| = 1$ we get that $G[U_i]$ is connected for all $i \in \{1,2,3\}$.
 Also, $E_G(U_i,\{w\}) \neq \emptyset$ for all $w \in \chi^{-1}(c)$ and $i \in \{1,2,3\}$.
 But this provides a minor isomorphic to $K_{3,h}$ with vertices $U_1,U_2,U_3$ on the left side and the vertices from $\chi^{-1}(c)$ on the right side. 
\end{proof}

\section{A Decomposition Theorem}
\label{sec:decomposition}

In the following, we use the insights gained in the last section to prove a decomposition theorem for graphs that exclude $K_{3,h}$ as a minor.
In the remainder of this work, all tree decompositions are rooted, i.e., there is a designated root node and we generally assume all edges to be directed away from the root.

\begin{theorem}
 \label{thm:decomposition-into-2-2-bounded-parts}
 Suppose $h \geq 3$.
 Let $G$ be a $3$-connected graph, and suppose $S \subseteq V(G)$ such that
 \begin{enumerate}[label = (\Alph*)]
  \item $G - E(S,S)$ excludes $K_{3,h}$ as a minor,
  \item $3 \leq |S| \leq h$,
  \item $G - S$ is connected, and
  \item $S = N_G(V(G) \setminus S)$.
 \end{enumerate}
 Then there is a (rooted) tree decomposition $(T,\beta)$ of $G$, a function $\gamma\colon V(T) \rightarrow 2^{V(G)}$, and a vertex-coloring $\lambda$ such that
 \begin{enumerate}[label = (\Roman*)]
  \item\label{item:decomposition-output-1} $|V(T)| \leq 2 \cdot |V(G)|$,
  \item\label{item:decomposition-output-2} the adhesion width of $(T,\beta)$ is at most $h-1$,
  \item\label{item:decomposition-output-3} for every $t \in V(T)$ with children $t_1,\dots,t_\ell$, one of the following options holds:
  \begin{enumerate}[label = (\alph*)]
   \item $\beta(t) \cap \beta(t_i) \neq \beta(t) \cap \beta(t_j)$ for all distinct $i,j \in [\ell]$, or
   \item $\beta(t) = \beta(t) \cap \beta(t_i)$ for all $i \in [\ell]$,
  \end{enumerate}
  \item\label{item:decomposition-output-4} $S \subsetneq \gamma(r)$ where $r$ denotes the root of $T$,
  \item\label{item:decomposition-output-5} $|\gamma(t)| \leq h^4$ for every $t \in V(T)$,
  \item\label{item:decomposition-output-6} $\beta(t) \cap \beta(s) \subseteq \gamma(t) \subseteq \beta(t)$ for all $t \in V(T) \setminus \{r\}$, where $s$ denotes the parent of $t$, and
  \item\label{item:decomposition-output-7} $\beta(t) \subseteq \cl_{2,2}^{(G,\lambda)}(\gamma(t))$ for all $t \in V(T)$.
 \end{enumerate}
 Moreover, the decomposition $(T,\beta)$, the function $\gamma$, and the coloring $\lambda$ can be computed in polynomial time, and the output is isomorphism-invariant with respect to $(G,S,h)$.
\end{theorem}

\begin{proof}
 We give an inductive construction for the tree decomposition $(T,\beta)$ as well as the function $\gamma$ and the coloring $\lambda$.
 We start by arguing how to compute the set $\gamma(r)$.
 
 \begin{claim}
  Let $v_1,v_2,v_3 \in S$ be three distinct vertices, and define
  \[\chi\coloneqq \WL{1}{G,S,v_1,v_2,v_3}.\]
  Then there exists some $v \in V(G) \setminus S$ such that $|[v]_\chi| < h$.
 \end{claim}
 \begin{claimproof}
  Let $H \coloneqq (G - (S \setminus \{v_1,v_2,v_3\})) - E(\{v_1,v_2,v_3\},\{v_1,v_2,v_3\})$.
  It is easy to see that $\chi|_{V(H)}$ is $1$-stable for the graph $H$.
  Observe that $H - \{v_1,v_2,v_3\} = G - S$ is connected.
  Suppose there is no vertex $v \in V(G) \setminus S$ such that $|[v]_\chi| < h$.
  Then $\chi$ is $(h-1)$-CR-stable which implies that $\cl_{h-1,1}^G(v_1,v_2,v_3) = \{v_1,v_2,v_3\}$.
  On the other hand, $Z \coloneqq V(H) \setminus \{v_1,v_2,v_3\}$ induces a connected component of $H - \{v_1,v_2,v_3\}$, and $N_H(Z) = \{v_1,v_2,v_3\}$ since $S = N_G(V(G) \setminus S)$.
  But this contradicts Lemma \ref{la:find-small-color-class}.
 \end{claimproof}
 
 Let $v_1,v_2,v_3 \in S$ be distinct.
 We define $\chi[v_1,v_2,v_3] \coloneqq \WL{1}{G,S,v_1,v_2,v_3}$.
 Also, let $c[v_1,v_2,v_3]$ denote the unique color such that
 \begin{enumerate}
  \item $c[v_1,v_2,v_3] \notin \{\chi[v_1,v_2,v_3](v) \mid v \in S\}$, and
  \item $|(\chi[v_1,v_2,v_3])^{-1}(c[v_1,v_2,v_3])| \leq h-1$
 \end{enumerate}
 and which is minimal with respect to the linear order on the colors in the image of $\chi[v_1,v_2,v_3]$.
 Let $\gamma(v_1,v_2,v_3) \coloneqq (\chi[v_1,v_2,v_3])^{-1}(c[v_1,v_2,v_3])$.
 Observe that $\gamma(v_1,v_2,v_3)$ is defined in an isomorphism-invariant manner given $(G,S,h,v_1,v_2,v_3)$.
 Now, define
 \[\gamma(r) \coloneqq S \cup \bigcup_{v_1,v_2,v_3 \in S \text{ distinct}} \gamma(v_1,v_2,v_3).\]
 Clearly, $\gamma(r)$ is defined in an isomorphism-invariant manner given $(G,S,h)$.
 Moreover,
 \[|\gamma(r)| \leq |S| + |S|^3 \cdot (h-1) \leq |S|^3 \cdot h \leq h^4.\]
 Finally, define $\beta(r) \coloneqq \cl_{2,2}^G(\gamma(r))$.
 
 Let $Z_1,\dots,Z_\ell$ be the connected components of $G - \beta(r)$.
 Also, let $S_i \coloneqq N_G(Z_i)$ and $G_i$ be the graph obtained from $G[S_i \cup Z_i]$ by turning $S_i$ into a clique, $i \in [\ell]$.
 We have $|S_i| < h$ by Lemma \ref{la:small-separator}.
 Also, $|S_i| \geq 3$ and $G_i$ is $3$-connected since $G$ is $3$-connected.
 Clearly, $G_i - S_i$ is connected and $S_i = N_{G_i}(V(G_i) \setminus S_i)$.
 Finally, $G_i - E(S_i,S_i)$ excludes $K_{3,h}$ as a minor because $G - E(S,S)$ excludes $K_{3,h}$ as a minor.
 
 We wish to apply the induction hypothesis to the triples $(G_i,S_i,h)$.
 If $|V(G_i)| = |V(G)|$ then $\ell = 1$ and $S \subsetneq S_i$.
 In this case the algorithm still makes progress since the size of $S$ can be increased at most $h-3$ times.
 
 By the induction hypothesis, there are tree decompositions $(T_i,\beta_i)$ of $G_i$ and functions $\gamma_i\colon V(T_i) \rightarrow 2^{V(G_i)}$ satisfying Properties \ref{item:decomposition-output-1} - \ref{item:decomposition-output-7}.
 We define $(T,\beta)$ to be the tree decomposition where $T$ is obtained from the disjoint union of $T_1,\dots,T_\ell$ by adding a fresh root vertex $r$ which is connected to the root vertices of $T_1,\dots,T_\ell$.
 Also, $\beta(r)$ is defined as above and $\beta(t) \coloneqq \beta_i(t)$ for all $t \in V(T_i)$ and $i \in [\ell]$.
 Finally, $\gamma(r)$ is again defined as above, and $\gamma(t) \coloneqq \gamma_i(t)$ for all $t \in V(T_i)$ and $i \in [\ell]$.
 
 The algorithm clearly runs in polynomial time and the output is isomorphism-invariant (the coloring $\lambda$ is defined below).
 We need to verify that Properties \ref{item:decomposition-output-1} - \ref{item:decomposition-output-7} are satisfied.
 Using the comments above and the induction hypothesis, it is easy to verify that Properties \ref{item:decomposition-output-2}, \ref{item:decomposition-output-4}, \ref{item:decomposition-output-5} and \ref{item:decomposition-output-6} are satisfied.
 
 For Property \ref{item:decomposition-output-7} it suffices to ensure that $\cl_{2,2}^{(G_i,\lambda)}(\gamma(t)) \subseteq \cl_{2,2}^{(G,\lambda)}(\gamma(t))$.
 Towards this end, it suffices to ensure that $\lambda(v) \neq \lambda(w)$ for all $v \in \beta(r)$ and $w \in V(G) \setminus \beta(r)$.
 To ensure this property holds on all levels of the tree, we can simply define $\lambda(v) \coloneqq \{\dist_T(r,t) \mid t \in V(T), v \in \beta(t)\}$.
 
 Next, we modify the tree decomposition in order to ensure Property \ref{item:decomposition-output-3}.
 Consider a node $t \in V(T)$ with children $t_1,\dots,t_\ell$.
 We say that $t_i \sim t_j$ if $\beta(t) \cap \beta(t_i) = \beta(t) \cap \beta(t_j)$.
 Let $A_1,\dots,A_k$ be the equivalence classes of the equivalence relation $\sim$.
 For every $i \in [k]$ we introduce a fresh node $s_i$.
 Now, every $t_j \in A_i$ becomes a child of $s_i$ and $s_i$ becomes a child of $t$.
 Finally, we set $\beta(s_i) = \gamma(s_i) \coloneqq \beta(t) \cap \beta(t_j)$ for some $t_j \in A_i$.
 Observe that after this modification, Properties \ref{item:decomposition-output-2} and \ref{item:decomposition-output-4} - \ref{item:decomposition-output-7} still hold.
 
 Finally, it remains to verify Property \ref{item:decomposition-output-1}.
 Before the modification described in the last paragraph, we have that $|V(T)| \leq |V(G)|$.
 Since the modification process at most doubles the number of nodes in $T$, the bound follows.
\end{proof}

\section{An FPT Isomorphism Test for Graphs of Small Genus}
\label{sec:main-algorithm}

Building on the decomposition theorem given in the last section, we can now prove the main result of this paper.

\begin{theorem}
 Let $G_1,G_2$ be two (vertex- and arc-colored) graphs that exclude $K_{3,h}$ as a minor.
 Then one can decide whether $G_1$ is isomorphic to $G_2$ in time $2^{\CO(h^4 \log h)}n^{\CO(1)}$.
\end{theorem}

\begin{proof}
 Suppose $G_i = (V(G_i),E(G_i),\chi_V^i,\chi_E^i)$ for $i \in \{1,2\}$.
 Using standard reduction techniques (see, e.g., \cite{HopcroftT72}) we may assume without loss of generality that $G_1$ and $G_2$ are $3$-connected.
 Pick an arbitrary set $S_1 \subseteq V(G_1)$ such that $|S_1| = 3$ and $G_1 - S_1$ is connected.
 For every $S_2 \subseteq V(G_2)$ such that $|S_2| = 3$ and $G_2 - S_2$ is connected, the algorithm tests whether there is an isomorphism $\varphi\colon G_1 \cong G_2$ such that $S_1^\varphi = S_2$.
 Observe that $S_i = N_{G_i}(V(G_i) \setminus S_i)$ for both $i \in \{1,2\}$ since $G_1$ and $G_2$ are $3$-connected.
 This implies that the triple $(G_i,S_i,h)$ satisfies the requirements of Theorem \ref{thm:decomposition-into-2-2-bounded-parts}.
 Let $(T_i,\beta_i)$ be the tree decomposition, $\gamma_i\colon V(T_i) \rightarrow 2^{V(G_i)}$ be the function, and $\lambda_i$ be the vertex-coloring computed by Theorem \ref{thm:decomposition-into-2-2-bounded-parts} on input $(G_i,S_i,h)$.
 
 Now, the basic idea is compute isomorphisms between $(G_1,S_1)$ and $(G_2,S_2)$ using dynamic programming along the tree decompositions.
 More precisely, we aim at recursively computing the set
 \[\Lambda \coloneqq \Iso((G_1,\lambda_1,S_1),(G_2,\lambda_1,S_2))[S_1]\]
 (here, $\Iso((G_1,\lambda_1,S_1),(G_2,\lambda_1,S_2))$ denotes the set of isomorphisms $\varphi\colon G_1 \cong G_2$ which additionally respect the vertex-colorings $\lambda_i$ and satisfy $S_1^\varphi = S_2$).
 Throughout the recursive algorithm, we maintain the property that $|S_i| \leq h$.
 Also, we may assume without loss of generality that $S_i$ is $\lambda_i$-invariant (otherwise, we replace $\lambda_i$ by $\lambda_i'$ defined via $\lambda_i'(v) \coloneqq (1,\lambda_i(v))$ for all $v \in S_i$, and $\lambda_i'(v) \coloneqq (0,\lambda_i(v))$ for all $v \in V(G_i) \setminus S_i$).
 
 Let $r_i$ denote the root node of $T_i$.
 Let $\ell$ denote the number of children of $r_i$ in the tree $T_i$ (if the number of children is not the same, the algorithm concludes that $\Iso((G_1,\lambda_1,S_1),(G_2,\lambda_1,S_2)) = \emptyset$).
 Let $t_1^i,\dots,t_\ell^i$ be the children of $r_i$ in $T_i$, $i \in \{1,2\}$.
 For $i \in \{1,2\}$ and $j \in [\ell]$ let $V_j^i$ denote the set of vertices appearing in bags below (and including) $t_j^i$.
 Also let $S_j^i \coloneqq \beta_i(r_i) \cap \beta_i(t_j^i)$ be the adhesion set to the $j$-th child, and define $G_{j}^i \coloneqq G_i[V_j^i]$.
 Finally, let $T_j^i$ denote the subtree of $T_i$ rooted at node $t_j^i$, and $\beta_j^i \coloneqq \beta_i|_{V(T_j^i)}$, $\gamma_j^i \coloneqq \gamma_i|_{V(T_j^i)}$ and $\lambda_j^i \coloneqq \lambda_i|_{V_j^i}$.
 
 For every $i,i' \in \{1,2\}$, and every $j,j' \in [\ell]$, the algorithm recursively computes the set
 \[\Lambda_{j,j'}^{i,i'} \coloneqq \Iso((G_j^i,\lambda_j^i,S_j^i),(G_{j'}^{i'},\lambda_{j'}^{i'},S_{j'}^{i'}))[S_j^i].\]
 We argue how to compute the set $\Lambda$.
 
 Building on Theorem \ref{thm:decomposition-into-2-2-bounded-parts}, Item \ref{item:decomposition-output-3}, we may assume that
 \begin{enumerate}[label = (\alph*)]
  \item\label{item:option-all-adhesion-sets-distinct} $S_j^i \neq S_{j'}^i$ for all distinct $j,j' \in [\ell]$ and $i \in \{1,2\}$, or
  \item\label{item:option-all-adhesion-sets-equal} $\beta(r_i) = S_j^i$ for all $j \in [\ell]$ and $i \in \{1,2\}$
 \end{enumerate}
 (if $r_1$ and $r_2$ do not satisfy the same option, then $\Iso((G_1,\lambda_1,S_1),(G_2,\lambda_1,S_2)) = \emptyset$).
 
 We first cover Option \ref{item:option-all-adhesion-sets-equal}.
 In this case $|\beta(r_i)| = |S_j^i| \leq h-1$ by Theorem \ref{thm:decomposition-into-2-2-bounded-parts}, Item \ref{item:decomposition-output-2}.
 The algorithm iterates over all bijections $\sigma \colon \beta(r_1) \rightarrow \beta(r_2)$.
 Now,
 \[\sigma \in \Iso((G_1,\lambda_1,S_1),(G_2,\lambda_1,S_2))[\beta(r_1)] \;\;\;\Leftrightarrow\;\;\; \exists \rho \in \Sym([\ell])\; \forall j \in [\ell]\colon \sigma \in \Lambda_{j,\rho(j)}^{1,2}.\]
 To test whether $\sigma$ satisfies the right-hand side condition, the algorithm constructs an auxiliary graph $H_\sigma$ with vertex set $V(H_\sigma) \coloneqq \{1,2\} \times [\ell]$ and edge set
 \[E(H_\sigma) \coloneqq \{(1,j)(2,j') \mid \sigma \in \Lambda_{j,j'}^{1,2}\}.\]
 Observe that $H_\sigma$ is bipartite with bipartition $(\{1\} \times [\ell], \{2\} \times [\ell])$.
 Now,
 \[\sigma \in \Iso((G_1,\lambda_1,S_1),(G_2,\lambda_1,S_2))[\beta(r_1)] \;\;\;\Leftrightarrow\;\;\; H_\sigma \text{ has a perfect matching}.\]
 It is well-known that the latter can be checked in polynomial time.
 This completes the description of the algorithm in case Option \ref{item:option-all-adhesion-sets-equal} is satisfied.
 
 Next, suppose Option \ref{item:option-all-adhesion-sets-distinct} is satisfied.
 Here, the central idea is to construct auxiliary vertex- and arc-colored graphs $H_i = (V(H_i),E(H_i),\mu_V^i,\mu_E^i)$ and sets $A_i \subseteq V(H_i)$ such that
 \begin{enumerate}
  \item $\beta_i(r_i) \subseteq A_i$ and $A_i \subseteq \cl_{2,2}^{H_i}(\gamma_i(r_i))$, and
  \item $\Iso(H_1,H_2)[S_1] = \Iso(H_1[A_1],H_2[A_2])[S_1] = \Lambda$.
 \end{enumerate}
 Towards this end, we set
 \[V(H_i) \coloneqq V(G_i) \uplus \{(S_j^i,\gamma) \mid j \in [\ell], \gamma \in \Lambda_{j,j}^{i,i}\}\]
 and
 \[E(H_i) \coloneqq E(G_i) \cup \{(S_j^i,\gamma)v \mid j \in [\ell], \gamma \in \Lambda_{j,j}^{i,i}, v \in S_j^i\}.\]
 Also, we set
 \[A_i \coloneqq \beta(r_i) \cup \{(S_j^i,\gamma) \mid j \in [\ell], \gamma \in \Lambda_{j,j}^{i,i}\}.\]
 The main idea is to use the additional vertices attached to the set $S_j^i$ to encode the isomorphism type of the graph $(G_j^i,\lambda_j^i,S_j^i)$.
 This information is encoded by the vertex- and arc-coloring building on sets $\Lambda_{j,j'}^{i,i'}$ already computed above.
 Let $\CS \coloneqq \{S_j^i \mid i \in \{1,2\}, j \in [\ell]\}$, and define $S_j^i \sim S_{j'}^{i'}$ if $\Lambda_{j,j'}^{i,i'} \neq \emptyset$.
 Observe that $\sim$ is an equivalence relation.
 Let $\{\CP_1,\dots,\CP_k\}$ be the partition of $\CS$ into the equivalence classes. 
 We set
 \[\mu_V^i(v) \coloneqq (0,\chi_V^i(v),\lambda_i(v))\]
 for all $v \in S_i$,
 \[\mu_V^i(v) \coloneqq (1,\chi_V^i(v),\lambda_i(v))\]
 for all $v \in \gamma_i(r_i) \setminus S_i$,
 \[\mu_V^i(v) \coloneqq (2,\chi_V^i(v),\lambda_i(v))\]
 for all $v \in \beta_i(r_i) \setminus \gamma_i(r_i)$,
 \[\mu_V^i(v) \coloneqq (3,\chi_V^i(v),\lambda_i(v))\]
 for all $v \in V(G_i) \setminus \beta_i(r_i)$, and
 \[\mu_V^i(S_j^i,\gamma) \coloneqq (4,q,q)\]
 for all $q \in [k]$, $S_j^i \in \CP_q$, and $\gamma \in \Lambda_{j,j}^{i,i}$.
 For every $q \in [k]$ fix some $i(q) \in \{1,2\}$ and $j(q) \in [\ell]$ such that $S_{j(q)}^{i(q)} \in \CP_q$ (i.e., for each equivalence class, the algorithm fixes one representative).
 Also, for every $q \in [k]$ and $S_j^i \in \CP_q$, fix a bijection $\sigma_j^i \in \Lambda_{j(q),j}^{i(q),i}$ such that $\sigma_{j(q)}^{i(q)}$ is the identity mapping.
 Finally, for $q \in [k]$, fix a numbering $S_{j(q)}^{i(q)} = \{u_1^q,\dots,u_{s(q)}^q\}$.

 With this, we are ready to define the arc-coloring $\mu_E^i$.
 First, we set
 \[\mu_E^i(v,w) \coloneqq (0,\chi_E^i(v,w))\]
 for all $vw \in E(G_i)$.
 Next, consider an edge $(S_j^i,\gamma)v$ where $j \in [\ell]$, $\gamma \in \Lambda_{j,j}^{i,i}$, and $v \in S_j^i$.
 Suppose $S_j^i \in \CP_q$.
 We set
 \[\mu_E^i(v,(S_j^i,\gamma)) = \mu_E^i((S_j^i,\gamma),v) \coloneqq (1,c)\]
 for the unique $c \in [s(q)]$ such that
 \[v = (u_c^q)^{\sigma_j^i\gamma}.\]
 
 This completes the description of the graphs $H_i$ and the sets $A_i$, $i \in \{1,2\}$.
 Next, we verify that they indeed have the desired properties.
 
 \begin{claim}
  \label{claim:closure-contains-root-bag}
  $\beta_i(r_i) \subseteq A_i$ and $A_i \subseteq \cl_{2,2}^{H_i}(\gamma_i(r_i))$.
 \end{claim}
 \begin{claimproof}
  The first part holds by definition of the set $A_i$.
  So let us verify the second part.
  We have that $V(G_i) \subseteq V(H_i)$ by definition.
  Also, $V(G_i)$ is $\mu_V^i$-invariant and $(\mu_V^i)[V(G_i)] \preceq \lambda_i$.
  This implies that
  \[\beta_i(r_i) \subseteq \cl_{2,2}^{(G_i,\lambda_i)}(\gamma(r_i)) \subseteq \cl_{2,2}^{H_i}(\gamma_i(r_i))\]
  using Theorem \ref{thm:decomposition-into-2-2-bounded-parts}, Item \ref{item:decomposition-output-7}.
  By definition of the set $A_i$, it remains to prove that
  \[\{(S_j^i,\gamma) \mid j \in [\ell], \gamma \in \Lambda_{j,j}^{i,i}\} \subseteq \cl_{2,2}^{H_i}(\gamma_i(r_i)).\]
  Actually, it suffices to prove that
  \[\{(S_j^i,\gamma) \mid j \in [\ell], \gamma \in \Lambda_{j,j}^{i,i}\} \subseteq \cl_{2,2}^{H_i}(\beta_i(r_i)).\]
  Let $\chi_i^*$ denote the $(2,2)$-WL-stable coloring after individualizing all vertices from $\beta_i(r_i)$.
  We have that $N_{H_i}(S_j^i,\gamma) = S_j^i \subseteq \beta_i(r_i)$ for all $j \in [\ell]$ and $\gamma \in \Lambda_{j,j}^{i,i}$.
  Recall that $S_j^i \neq S_{j'}^i$ for all distinct $j,j' \in [\ell]$.
  This implies that $\chi_i^*(S_j^i,\gamma) \neq \chi_i^*(S_{j'}^i,\gamma')$ for all distinct $j,j' \in [\ell]$, $\gamma \in \Lambda_{j,j}^{i,i}$ and $\gamma' \in \Lambda_{j',j'}^{i,i}$.
  So pick $\gamma,\gamma' \in \Lambda_{j,j}^{i,i}$ to be distinct.
  We need to show that $\chi_i^*(S_j^i,\gamma) \neq \chi_i^*(S_j^i,\gamma')$.
  Here, we use the arc-coloring $\mu_E^i$.
  Indeed, by definition, all edges incident to the vertex $(S_j^i,\gamma)$ receive pairwise different colors.
  Now pick $w \in S_j^i$ such that $w^{\gamma} \neq w^{\gamma'}$ and suppose that $w = (u_c^q)^{\sigma_j^i}$ where $q \in [k]$ is the unique number for which $S_j^i \in \CP_q$.
  Then $(S_j^i,\gamma)$ and $(S_j^i,\gamma')$ reach different neighbors via $(1,c)$-colored edges.
  Since all neighbors of $(S_j^i,\gamma)$ and $(S_j^i,\gamma')$ are already individualized in $\chi_i^*$, we conclude that $\chi_i^*(S_j^i,\gamma) \neq \chi_i^*(S_j^i,\gamma')$.
  This implies that 
  \[\{(S_j^i,\gamma) \mid j \in [\ell], \gamma \in \Lambda_{j,j}^{i,i}\} \subseteq \cl_{2,2}^{H_i}(\beta_i(r_i))\]
  as desired.
 \end{claimproof}
 
 \begin{claim}
  \label{claim:equal-isomorphism-sets}	
  \[\Iso(H_1,H_2)[S_1] = \Iso(H_1[A_1],H_2[A_2])[S_1] = \Iso((G_1,S_1),(G_2,S_2))[S_1].\]
 \end{claim}
 \begin{claimproof}
  We show the claim by proving the following three inclusions.
  \begin{description}
   \item[{$\Iso(H_1,H_2)[S_1] \subseteq \Iso(H_1[A_1],H_2[A_2])[S_1]$:}] Observe that $A_i$ is $\mu_V^i$-invariant. Hence, this inclusion trivially holds.
   \item[{$\Iso(H_1[A_1],H_2[A_2])[S_1] \subseteq \Iso((G_1,S_1),(G_2,S_2))[S_1]$:}] Pick $\psi \in \Iso(H_1[A_1],H_2[A_2])$.
    First observe that, by the vertex-coloring of $H_i$, we have that $S_1^\psi = S_2$.
    From the structure of $H_i[A_i]$ it follows that there is some permutation $\rho \in \Sym([\ell])$ such that, for every $j \in [\ell]$ and $\gamma \in \Lambda_{j,j}^{1,1}$ there is some $\gamma' \in \Lambda_{\rho(j),\rho(j)}^{2,2}$ such that
    \[\psi(S_j^1,\gamma) = (S_{\rho(j)}^2,\gamma')\]
    Moreover, $S_j^1 \sim S_{\rho(j)}^2$ for all $j \in [\ell]$.
    Now, fix some $j \in [\ell]$.
    We argue that $\psi[S_j^i] \in \Lambda_{j,\rho(j)}^{1,2}$.
    Towards this end, pick $q \in [k]$ such that $S_j^1 \in \CP_q$.
    Observe that also $S_{\rho(j)}^2 \in \CP_q$.
    
    Let $\gamma \in \Lambda_{j,j}^{1,1}$ be the identify mapping, and pick $\gamma' \in \Lambda_{\rho(j),\rho(j)}^{2,2}$ such that $\psi(S_j^1,\gamma) = (S_{\rho(j)}^2,\gamma')$.
    Let $v_1 \in S_j^1$ and pick $c \in [s(q)]$ such that
    \[\mu_E^1(v_1,(S_j^1,\gamma)) = (1,c).\]
    Also, let $v_2 \coloneqq \psi(v_1)$.
    Since $\psi$ is an isomorphism, we conclude that
    \[\mu_E^2(v_2,(S_j^2,\gamma')) = (1,c).\]
    By definition of the arc-colorings, this means that
    \[v_1 = (u_c^q)^{\sigma_j^1\gamma} = (u_c^q)^{\sigma_j^1}\]
    and
    \[v_2 = (u_c^q)^{\sigma_{\rho(j)}^2\gamma'}.\]
    Together, this means that
    \[v_2 = v_1^{(\sigma_j^1)^{-1}\sigma_{\rho(j)}^2\gamma'}.\]
    In particular, $\psi[S_j^i] = (\sigma_j^1)^{-1}\sigma_{\rho(j)}^2\gamma' \in \Lambda_{j,\rho(j)}^{1,2}$.
    
    So let
    \[\varphi_{j} \in \Iso((G_j^1,S_j^1),(G_{\rho(j)}^{2},S_{\rho(j)}^{2}))\]
    such that $\varphi_j[S_j^1] = \psi[S_j^1]$ for all $j \in [\ell]$.
    Finally, define $\varphi\colon V(G_1) \rightarrow V(G_2)$ via $\varphi(v) \coloneqq \psi(v)$ for all $v \in \beta(r_1)$, and $\varphi(v) \coloneqq \varphi_j(v)$ for all $v \in V_j^1$.
    It is easy to see that $\varphi \in \Iso((G_1,S_1),(G_2,S_2))$ and $\varphi[S_1] = \psi[S_1]$.
   \item[{$\Iso((G_1,S_1),(G_2,S_2))[S_1] \subseteq \Iso(H_1,H_2)[S_1]$:}]
    Let $\varphi \in \Iso((G_1,S_1),(G_2,S_2))$ be an isomorphism.
    We have that $(T_i,\beta_i)$, the function $\gamma_i$, and the coloring $\lambda_i$ are computed in an isomorphism-invariant fashion.
    Hence, $\varphi\colon H_1[V(G_1)] \cong H_2[V(G_2)]$.
    Also, there is a permutation $\rho \in \Sym([\ell])$ such that $(S_j^1)^\varphi = S_{\rho(j)}^2$.
    
    We now define a bijection $\psi\colon V(H_1) \rightarrow V(H_2)$ as follows.
    First of all, $\psi(v) \coloneqq \varphi(v)$ for all $v \in V(G_i) \subseteq V(H_i)$.
    Now, let $j \in [\ell]$ and $\gamma \in \Lambda_{j,j}^{1,1}$.
    Observe that $\varphi[S_j^1] \in \Lambda_{j,\rho(j)}^{1,2}$.
    Pick $q \in [k]$ such that $S_j^1,S_{\rho(j)}^2 \in \CP_q$.
    We have that
    \[\varphi[S_j^1] = \gamma^{-1}(\sigma_j^1)^{-1}\sigma_{\rho(j)}^2\gamma'\]
    for some $\gamma' \in \Lambda_{\rho(j),\rho(j)}^{2,2}$.
    We set
    \[\varphi(S_j^1,\gamma) \coloneqq (S_{\rho(j)}^2,\gamma').\]
    By definition, $\mu_V^1(S_j^1,\gamma) = (4,q,q) = \mu_V^2(S_{\rho(j)}^2,\gamma')$.
    It only remains to verify that, for all $v \in S_j^1$, we have that
    \[\mu_E^1((S_j^1,\gamma),v) = \mu_E^2((S_{\rho(j)}^2,\gamma'),\varphi(v)).\]
    Suppose that $\mu_E^1((S_j^1,\gamma),v) = (1,c)$.
    Then $v = (u_c^q)^{\sigma_j^1\gamma}$.
    On the other hand,
    \[(u_c^q)^{\sigma_{\rho(j)}^2\gamma'} = v^{\gamma^{-1}(\sigma_j^1)^{-1}\sigma_{\rho(j)}^2\gamma'} = v^{\varphi}.\]
    This implies that $\mu_E^2((S_{\rho(j)}^2,\gamma'),\varphi(v)) = (1,c)$.
    Overall, we get that $\psi \colon H_1 \cong H_2$.
  \end{description}
 \end{claimproof}

 Recall that the algorithm aims at computing the set $\Lambda$.
 Building on the previous claim, we can simply compute $\Iso(H_1[A_1],H_2[A_2])[S_1]$.
 Towards this end, the algorithm iterates through all bijections $\tau\colon\gamma_1(r_1)\rightarrow\gamma_2(r_2)$, and wishes to test whether there is an isomorphism $\varphi \in \Iso(H_1[A_1],H_2[A_2])$ such that $\varphi[\gamma_1(r_1)] = \tau$.
 Note that, since $\gamma_i(r_i)$ is $\mu_V^i$-invariant, it now suffices to solve this latter problem.
 
 So fix a bijection $\tau\colon\gamma_1(r_1)\rightarrow\gamma_2(r_2)$ (if $|\gamma_1(r_1)| \neq |\gamma_2(r_2)|$ then the algorithm returns $\Lambda = \emptyset$).
 Let $\mu_1^*(v) \coloneqq (1,v)$ for all $v \in \gamma_1(r_1)$, $\mu_1^*(v) \coloneqq (0,\mu_V^1)$ for all $v \in V(H_1) \setminus \gamma_1(r_1)$, and
 $\mu_2^*(v) \coloneqq (1,\tau^{-1}(v))$ for all $v \in \gamma_2(r_2)$ and $\mu_2^*(v) \coloneqq (0,\mu_V^2)$ for all $v \in V(H_2) \setminus \gamma_2(r_2)$.
 
 Intuitively speaking, $\mu_1^*$ and $\mu_2^*$ are obtained from $\mu_V^1$ and $\mu_V^2$ by individualizing all vertices from $\gamma_1(r_1)$ and $\gamma_r(r_2)$ according to the bijection $\tau$.
 Now, we can apply Theorem \ref{thm:bounding-group-t-k-wl} on input graph $H_1^* = (H_1,\mu_1^*)$ and $H_2^* = (H_2,\mu_2^*)$, and parameters $t = k \coloneqq 2$.
 
 Building on Claim \ref{claim:closure-contains-root-bag}, we obtain a $\mgamma_2$-group $\Gamma \leq \Sym(A_1)$ and a bijection $\theta\colon A_1 \rightarrow A_2$ such that $\Iso(H_1^*[A_1],H_2^*[A_2]) \subseteq \Gamma\theta$.
 Now, we can determine whether $H_1^*[A_1] \cong H_2^*[A_2]$ using Theorem \ref{thm:hypergraph-isomorphism-gamma-d}.
 Using Claim \ref{claim:equal-isomorphism-sets}, this provides the answer to whether $\tau[S_1] \in \Lambda$ (recall that $S_1 \subseteq \gamma_1(r_1)$ by Theorem \ref{thm:decomposition-into-2-2-bounded-parts}, Items \ref{item:decomposition-output-4} and \ref{item:decomposition-output-6}).
 
 Overall, this completes the description of the algorithm.
 It only remains to analyze its running time.
 Let $n$ denote the number of vertices of $G_1$ and $G_2$.
 
 The algorithm iterates over at most $n^3$ choices for the initial set $S_2$, and computes the decompositions $(T_i,\beta_i)$, the functions $\gamma_i$, and the colorings $\lambda_i$ in polynomial time.
 For the dynamic programming tables, the algorithm needs to compute $\CO(n^2)$ many $\Lambda$-sets (using Theorem \ref{thm:decomposition-into-2-2-bounded-parts}, Item \ref{item:decomposition-output-1}), each of which contains at most $h! = 2^{\CO(h \log h)}$ many elements by Theorem \ref{thm:decomposition-into-2-2-bounded-parts}, Item \ref{item:decomposition-output-2}.
 Hence, it remains to analyze the time required to compute the set $\Lambda$ given the $\Lambda_{j,j'}^{i,i'}$-sets.
 For Option \ref{item:option-all-adhesion-sets-equal}, this can clearly be done in time $2^{\CO(h \log h)}n^{\CO(1)}$.
 
 So consider Option \ref{item:option-all-adhesion-sets-distinct}.
 The graph $H_i$ can clearly be computed in time polynomial in its size.
 We have that $|V(H_i)| = 2^{\CO(h \log h)}n$.
 Afterwards, the algorithm iterates over $|\gamma_1(r_1)|!$ many bijections $\tau$.
 By Theorem \ref{thm:decomposition-into-2-2-bounded-parts}, Item \ref{item:decomposition-output-5}, we have that $|\gamma_1(r_1)|! = 2^{\CO(h^4 \log h)}$.
 For each bijection, the algorithm then requires polynomial computation time by Theorems \ref{thm:bounding-group-t-k-wl} and \ref{thm:hypergraph-isomorphism-gamma-d}.
 Overall, this proves the bound on the running time.
\end{proof}

\begin{remark}
 The algorithm from the last theorem can be extended in two directions.
 First, if one of the input graphs does not exclude $K_{3,h}$ as a minor, it can be modified to either correctly conclude that $G_1$ has a minor isomorphic to $K_{3,h}$, or to correctly decide whether $G_1$ is isomorphic to $G_2$.
 Indeed, the only part of the algorithm that exploits that the input graphs do not have minor isomorphic to $K_{3,h}$ is the computation of the tree decompositions $(T_i,\beta_i)$ from Theorem \ref{thm:decomposition-into-2-2-bounded-parts}.
 In turn, this theorem only exploits forbidden minors via Lemmas \ref{la:small-separator} and \ref{la:find-small-color-class}.
 An algorithm can easily detect if one of the implications of those two statements is violated, in which case it can infer the existence of a minor $K_{3,h}$.
 
 Secondly, using standard reduction techniques (see, e.g., \cite{Mathon79}), one can also compute a representation of the set of all isomorphisms $\Iso(G_1,G_2)$ in the same time.
\end{remark}

Since every graph $G$ of Euler genus $g$ excludes $K_{3,4g+3}$ as a minor \cite{Ringel65}, we obtain the following corollary.

\begin{corollary}
 Let $G_1,G_2$ be two (vertex- and arc-colored) graphs of Euler genus at most $g$.
 Then one can decide whether $G_1$ is isomorphic to $G_2$ in time $2^{\CO(g^4 \log g)}n^{\CO(1)}$.
\end{corollary}

\section{Conclusion}

We presented an isomorphism test for graphs excluding $K_{3,h}$ as a minor running in time $2^{\CO(h^4 \log h)}n^{\CO(1)}$.
For this, we provided a polynomial-time isomorphism algorithm for $(t,k)$-WL-bounded graphs and argued that graphs excluding $K_{3,h}$ as a minor can be decomposed into parts that are $(2,2)$-WL-bounded after individualizing a small number of vertices.

Still, several questions remain open.
Maybe most notable, can isomorphism testing for graphs excluding $K_h$ as a minor be tested in time $f(h)n^{O(1)}$ for some \emph{computable} function $f$.
As graphs of bounded genus form an important subclass of graphs excluding $K_h$ as a minor, the techniques developed in this paper may provide an alternative approach to this question compared to the algorithm from \cite{LokshtanovPPS22}.

As an intermediate problem, it might also be interesting to consider the class $\CG_h$ of all graphs $G$ for which there is a set $X \subseteq V(G)$ of size $|X| \leq h$ such that $G - X$ is planar.
Is there a simple and efficient fpt isomorphism test for the class $\CG_h$ parameterized by $h$?

\bibliographystyle{plainurl}
\small
\bibliography{literature}

\end{document}